\documentclass[a4paper,UKenglish,cleveref,autoref,11pt]{article}

\usepackage{a4}
\usepackage{authblk}
\usepackage{lineno,hyperref}
\usepackage{amsmath}
\usepackage{amsthm}
\usepackage{cases}
\usepackage{booktabs}
\usepackage{microtype}
\usepackage{multirow}
\usepackage[pdftex]{graphicx}
\usepackage[T1]{fontenc}
\usepackage[final,mode=multiuser]{fixme}
\usepackage{geometry}
\usepackage[whole]{bxcjkjatype}

\theoremstyle{definition}

\newtheorem{theorem}{Theorem}
\newtheorem{lemma}{Lemma}
\newtheorem{corollary}{Corollary}

\newtheorem{remark}{Remark}

\fxuseenvlayout{colorsig}
\fxusetargetlayout{color}

\FXRegisterAuthor{mf}{amf}{MF}
\FXRegisterAuthor{yn}{ayn}{YN}
\FXRegisterAuthor{si}{asi}{SI}
\FXRegisterAuthor{hb}{ahb}{HB}
\FXRegisterAuthor{mt}{amt}{MT}
\fxsetface{env}{}

\title{
  
  Computing longest palindromic substring after single-character or block-wise edits\footnote{%
    Preliminary versions of this work appeared in~\cite{Funakoshi_CPM2018,Funakoshi_CPM2019}.}}

\author[1]{Mitsuru Funakoshi}
\author[1]{Yuto Nakashima}
\author[1,2]{Shunsuke Inenaga}
\author[3]{Hideo~Bannai}
\author[1]{Masayuki Takeda}

\affil[1]{Department of Informatics, Kyushu University

\{mitsuru.funakoshi, yuto.nakashima, inenaga, takeda\}@inf.kyushu-u.ac.jp}
\affil[2]{PRESTO, Japan Science and Technology Agency}
\affil[3]{M\&D Data Science Center, Tokyo Medical and Dental University
hdbn.dsc@tmd.ac.jp}
\date{}

\newcommand{\PalPre}{\mathit{PrePals}}
\newcommand{\PalSuf}{\mathit{SufPals}}
\newcommand{\MaxPalE}{\mathit{MaxPalEnd}}
\newcommand{\MaxPalB}{\mathit{MaxPalBeg}}
\newcommand{\OutLCE}{\mathsf{OutLCE}}
\newcommand{\LeftLCE}{\mathsf{LeftLCE}}
\newcommand{\RightLCE}{\mathsf{RightLCE}}
\newcommand{\RLFbeg}{\mathit{RLFBeg}}
\newcommand{\RLFend}{\mathit{RLFEnd}}
\newcommand{\lcp}{\mathit{lcp}}
\newcommand{\Ext}{\mathit{Ext}}
\newcommand{\rev}[1]{{#1}^{\mathit{R}}}

\newcommand{\argmax}{\mathop{\rm arg~max}\limits}

\makeatletter
\newcommand{\figcaption}[1]{\def\@captype{figure}\caption{#1}}
\newcommand{\tblcaption}[1]{\def\@captype{table}\caption{#1}}
\makeatother

\newtheorem{observation}{Observation}

\begin{document}

\maketitle

\begin{abstract}
  \emph{Palindromes} are important objects in strings which have been
  extensively studied from combinatorial, algorithmic, and bioinformatics points of views.
  It is known that the length of the longest palindromic substrings (LPSs)
  of a given string $T$ of length $n$ can be computed
  in $O(n)$ time by Manacher's algorithm [J. ACM '75].
  In this paper, we consider the problem of finding
  the LPS after the string is edited.
  We present an algorithm that uses $O(n)$ time and space
  for preprocessing, and answers the length of the LPSs in
  $O(\log (\min \{\sigma, \log n\}))$ time after a single
  character substitution, insertion, or deletion,
  where $\sigma$ denotes the number of distinct characters appearing in $T$.
  We also propose an algorithm that uses $O(n)$ time and space
  for preprocessing, and answers the length of the LPSs
  in $O(\ell + \log \log n)$ time,
  after an existing substring in $T$ is replaced by a
  string of arbitrary length $\ell$.
\end{abstract}

\setcounter{footnote}{0}

\section{Introduction}

\emph{Palindromes} are strings that read the same forward and backward.
The problems of finding palindromes or palindrome-like structures
in a given string are fundamental tasks in string processing,
and thus have been extensively studied (e.g., see~\cite{Apostolico1995parallel,matsubara_tcs2009,DBLP:journals/ipl/GroultPR10,DBLP:conf/stringology/KosolobovRS13,Porto2002ApproxPalindrome,KolpakovK09,NarisadaDNIS17,GawrychowskiIIK18,DBLP:journals/ijfcs/AdamczykACR18} and references therein).

One of the earliest problems regarding palindromes is
the \emph{longest palindromic substring} (\emph{LPS}) problem,
which asks to find (the length) of the longest palindromes
that appear in a given string.
This problem dates back to 1970's~\cite{Manacher75},
and since then it has been popular as a good algorithmic exercise.
Observe that the longest palindromic substring is also
a maximal (non-extensible) palindrome in the string,
whose center is an integer position if its length is odd,
or a half-integer position if its length is even.
Since one can compute the maximal palindromes
for all such centers in $O(n^2)$ total time by na\"ive character comparisons,
the LPS problem can also be easily solved in $O(n^2)$ time.

Manacher~\cite{Manacher75} gave
an elegant $O(n)$-time solution to the LPS problem.
Manacher's algorithm uses symmetry of palindromes
and character
equality comparisons only,
and therefore works in $O(n)$ time for any alphabet.
It was pointed out in~\cite{Apostolico1995parallel}
that Manacher's algorithm actually computes all the maximal palindromes
in the string.
In the case where the input string is drawn from a constant size alphabet or
an integer alphabet of size polynomial in $n$,
there is an alternative suffix tree~\cite{Weiner73} based algorithm
which takes $O(n)$ time~\cite{gusfield97:_algor_strin_trees_sequen}.
This suffix-tree based algorithm also computes all maximal palindromes.
For a constant size alphabet or an integer alphabet, by using eertree~\cite{eertree}, which is a $O(n)$ space data structure representing all palindromes of a string, one can get LPSs in $O(n)$ time.
This data structure can be constructed in $O(n)$ time offline.
The LPS problem in the streaming model
has also been considered in the literature~\cite{BerenbrinkEMA14,GawrychowskiMSU16}.

Getting back to the problem of computing all maximal palindromes,
there is a simple $O(n)$-space data structure
representing all of the computed maximal palindromes;
simply store their lengths in an array of length $2n-1$
together with the input string $T$.
However, this data structure is apparently not flexible for string edits,
since even a single character substitution, insertion, or deletion
can significantly break palindromic structures of the string.
Indeed, $\Omega(n^2)$ palindromic substrings and
$\Omega(n)$ maximal palindromes can be affected by a single edit operation
(E.g., consider to replace the middle character of string $a^n$
with another character $b$).
Hence, an intriguing question is whether there exists
a space-efficient data structure for the input string $T$
which can quickly answer the following query:
What is the length of the longest palindromic substring(s),
if single character substitution, insertion, or deletion is performed?
We call this a \emph{1-ELPS query}.

We present an algorithm which uses $O(n)$ time and space
for preprocessing,
and $O(\log (\min \{\sigma, \log n\}))$ time for 1-ELPS queries,
where $\sigma$ is the number of distinct characters appearing in $T$.
Thus, our query algorithm runs in optimal $O(1)$ time for any constant-size alphabet,
and runs in $O(\log \log n)$ time for
larger alphabets of size $\sigma = \Omega(\log n)$.
In addition, this algorithm can readily be extended to
a randomized version with hashing,
which answers queries in $O(1)$ time each
and uses $O(n)$ expected time and $O(n)$ space for preprocessing.

We also consider a more general variant of 1-ELPS queries
which allows for a \emph{block-wise} edit operation,
where an existing substring in the input string $T$
can be replaced with a string of arbitrary length $\ell$,
called an \emph{$\ell$-ELPS queries}.
We present an algorithm which uses $O(n)$ time and space
for preprocessing and $O(\ell + \log \log n)$ time for $\ell$-ELPS queries.
We emphasize that this query time is independent of the length of the
original block (substring) to be edited.

All the results reported in this paper are valid for any string of length $n$
over an integer alphabet of size polynomial in $n$.

\subsection*{Related Work}
Amir et al.~\cite{AmirCIPR17} proposed an algorithm
to find the \emph{longest common factor} (\emph{LCF}) of two strings,
after a single character edit operation is performed in one of the strings.
Their data structure occupies $O(n \log^3 n)$ space and uses $O(\log^3 n)$ query time,
where $n$ is the length of the input strings.
Their data structure can be constructed in $O(n \log^4 n)$ expected time.
After that, Abedin et al.~\cite{AbedinH0T18} showed that the above problem can be reduced to the heaviest induced ancestors problem over two trees and solved in $O(\log n \log \log n)$ (resp. $O(\log^2 n \log \log n)$) query time using an $O(n \log n)$ (resp. $O(n)$) space data structure.
Urabe et al.~\cite{UrabeNIBT18} considered the problem of computing
the longest \emph{Lyndon word} in a string after an edit operation.
They showed algorithms for $O(\log n)$-time queries for a single character edit operation
and $O(\ell \log \sigma + \log n)$-time queries for a block-wise edit operation,
both using $O(n)$ time and space for preprocessing.
We note that in these results including ours in this current paper,
edit operations are given as \emph{queries} and thus the input string(s)
remain static even after each query.
This is due to the fact that changing the data structure dynamically can be too costly in many cases.

It is noteworthy, however, that
recently Amir et al.~\cite{DBLP:conf/esa/AmirCPR19}
solved dynamic versions for the LCF problem and some of its variants.
In particular, when $n$ is the maximum length of the string that can be edited,
they showed a data structure of $O(n \log n)$ space
that can be dynamically maintained and can answer $1$-ELPS queries
in $O(\sqrt{n}\log^{2.5}n)$ time,
after $O(n \log^2 n)$ time preprocessing.
Furthermore, Amir et al.~\cite{DBLP:journals/corr/abs-1906-09732} presented
an algorithm for computing the longest palindrome in a dynamic string in $O(\mathrm{polylog} \:n)$ time per single character substitution.
In comparison to these recent results on dynamic strings,
although our algorithm does not allow for changing the string,
our algorithm answers 1-ELPS queries in $O(\log \log n)$ time
or even faster for small alphabet size $\sigma$,
which is exponentially faster than $O(\mathrm{polylog} \:n)$.

In addition, Amir et al.~\cite{AmirBCK19} presented fully dynamic algorithm for maintaining a representation of the squares and Charalampopoulos et al.~\cite{Charalampopoulos20} showed polylogarithmic time algorithm for the LCF problem of two dynamic strings.

\section{Preliminaries}\label{sec:preliminaries}

\subsection{String Notations}

Let $\Sigma$ be the {\em alphabet}.
An element of $\Sigma^*$ is called a {\em string}.
The length of a string $T$ is denoted by $|T|$.
The empty string $\varepsilon$ is a string of length 0,
namely, $|\varepsilon| = 0$.
For a string $T = xyz$, $x$, $y$ and $z$ are called
a \emph{prefix}, \emph{substring}, and \emph{suffix} of $T$, respectively.
For two strings $X$ and $Y$,
let $\lcp(X, Y)$ denote the length of the longest
common prefix of $X$ and $Y$.
A prefix $x$ and a suffix $z$ of $T$ are respectively
called a \emph{proper prefix} and \emph{proper suffix} of $T$,
if $x \neq T$ and $z \neq T$.

For a string $T$ and an integer $1 \leq i \leq |T|$,
$T[i]$ denotes the $i$-th character of $T$,
and for two integers $1 \leq i \leq j \leq |T|$,
$T[i..j]$ denotes the substring of $T$
that begins at position $i$ and ends at position $j$.
For convenience, let $T[i..j] = \varepsilon$ when $i > j$.
An integer $p \geq 1$ is said to be a \emph{period}
of a string $T$ iff $T[i] = T[i+p]$ for all $1 \leq i \leq |T|-p$.
If a string $B$ is both a proper prefix and a proper suffix of another string $T$,
then $B$ is called a \emph{border} of $T$.

The \emph{run length} (\emph{RL}) factorization
of a string $T$ is
a sequence $f_1, \ldots, f_m$ of maximal runs of the same characters
such that $T = f_1 \cdots f_m$
(namely, each RL factor $f_j$ is a repetition of the same character $a_j$ with $a_j \neq a_{j+1}$).
For each position $1 \leq i \leq n$ in $T$,
let $\RLFbeg(i)$ and $\RLFend(i)$ denote the beginning and ending positions
of the RL factor that contains the position $i$, respectively.
One can easily compute in $O(n)$ time the RL factorization of
string $T$ of length $n$
together with $\RLFbeg(i)$ and $\RLFend(i)$ for all positions $1 \leq i \leq n$.

Let $\rev{T}$ denote the reversed string of $T$,
i.e., $\rev{T} = T[|T|] \cdots T[1]$.
A string $T$ is called a \emph{palindrome} if $T = \rev{T}$.
We remark that the empty string $\varepsilon$ is also
considered to be a palindrome.
A non-empty palindromic substring $T[i..j]$
is said to be a \emph{maximal palindrome} of $T$
if $T[i-1] \neq T[j+1]$, $i = 1$, or $j = |T|$.
For any non-empty palindromic substring $T[i..j]$ in $T$,
$\frac{i+j}{2}$ is called its \emph{center}.
It is clear that for each center $c = 1, 1.5, \ldots, n-0.5, n$,
we can identify the maximal palindrome $T[i..j]$ whose center is $c$
(namely, $c = \frac{i+j}{2}$).
Thus, there are exactly $2n-1$ maximal palindromes in a string of length $n$
(including empty ones which occur at non-integer center $c$ when $T[c-1/2] \neq T[c+1/2]$).
In particular, maximal palindromes $T[1..i]$ and $T[i..|T|]$ for $1 \leq i \leq n$
are respectively called a \emph{prefix palindrome} and a \emph{suffix palindrome} of $T$.

A \emph{rightward  longest common extension} (\emph{rightward  LCE})
query on a string $T$
is to compute $\lcp(T[i..|T|], T[j..|T|])$
for given two positions $1 \leq i \neq j \leq |T|$.
Similarly, a \emph{leftward LCE} query is
to compute $\lcp(\rev{T[1..i]}, \rev{T[1..j]})$.
We denote by $\RightLCE_T(i, j)$ and $\LeftLCE_T(i, j)$
rightward and leftward LCE queries for positions $1 \leq i \neq j \leq |T|$,
respectively.
An \emph{outward LCE} query is, given two positions
$1 \leq i < j \leq |T|$,
to compute $\lcp(\rev{(T[1..i])}, T[j..|T|])$.
We denote by $\OutLCE_{T}(i, j)$ an outward LCE query for positions $i < j$
in the string $T$.

\subsection{Computing Maximal Palindromes}

Manacher~\cite{Manacher75} showed an elegant online algorithm
which computes all maximal palindromes of a given string $T$ of length $n$
in $O(n)$ time.
An alternative offline approach is to use outward LCE queries
for $2n-1$ pairs of positions in $T$.
Using the suffix tree~\cite{Weiner73} for string $T\$\rev{T}\#$
enhanced with a lowest common ancestor data structure~\cite{HarelT84,SchieberV88,DBLP:conf/latin/BenderF00},
where $\$$ and $\#$ are special characters which do not appear in $T$,
each outward LCE query can be answered in $O(1)$ time.
For any integer alphabet of size polynomial in $n$,
preprocessing for this approach takes $O(n)$ time and space~\cite{Farach-ColtonFM00,gusfield97:_algor_strin_trees_sequen}.
Let $\mathcal{M}$ be an array of length $2n-1$
storing the lengths of maximal palindromes
in increasing order of centers.
For convenience, we allow the index for $\mathcal{M}$
to be an integer or a half-integer from $1$ to $n$,
so that $\mathcal{M}[i]$ stores the length of the maximal palindrome of
$T$ centered at $i$.

A palindromic substring $P$ of a string $T$
is called a \emph{longest palindromic substring} (\emph{LPS})
if there are no palindromic substrings of $T$ which are longer than $P$.
Since any LPS of $T$ is always a maximal palindrome of $T$,
we can find all LPSs and their lengths in $O(n)$ time.

\subsection{Our Problems}

In this paper, we consider the three standard edit operations,
i.e., insertion, deletion, and substitution of a character in the input string $T$ of length $n$.
Let $T'$ denote the string after one of the above edit operations was performed at a given position.
A \emph{1-edit longest palindromic substring} query
(\emph{1-ELPS} query) is to answer (the length of) a longest palindromic substring of $T'$.
In Section~\ref{sec:1-ELPS}, we will present an $O(n)$-time and space
preprocessing scheme such that
subsequent \emph{1-ELPS} queries can be answered
in $O(\log (\min \{\sigma, \log n\}))$ time.

For any integer $\ell \geq 0$,
an \emph{$\ell$-block edit longest palindromic substring} query
(\emph{$\ell$-ELPS} query), which is a generalization of
the 1-ELPS query, asks (the length of)
a longest palindromic substring of $T''$,
where $T''$ denotes the string after a substring of $T$ (of any length) is replaced by a string of length $\ell$.
In Section~\ref{sec:block-ELPS},
we will propose an $O(n)$-time and space preprocessing scheme
such that subsequent $\ell$-ELPS queries can be answered in $O(\ell + \log \log n)$ time.
We remark that in both problems
string edits are only given as \emph{queries}, i.e.,
we do not explicitly rewrite the original string $T$ into $T'$ nor $T''$
and $T$ remains unchanged for further queries.
We also remark that in our problem the length $\ell$ of a substring $X$
that substitutes a given interval (substring) can be arbitrary.

\subsection{Properties of Maximal Palindromes}

The following properties of palindromes are useful in our algorithms.

\begin{lemma} \label{lem:pal_border}
Any border $B$ of a palindrome $P$ is also a palindrome.
\end{lemma}

\begin{proof}
  Since $P$ is a palindrome, clearly $P[1..m] = \rev{(P[|P|-m+1..|P|])}$
  for any $1 \leq m \leq |P|$.
  Since $B$ is a border of $P$,
  $B = P[1..|B|] = \rev{(P[|P|-|B|+1..|P|])} = \rev{B}$.
\end{proof}

Let $T$ be a string of length $n$.
For each $1 \leq i \leq n$, let $\MaxPalE_T(i)$ denote the set of
maximal palindromes of $T$ that end at position $i$.
Let $\mathbf{S}_i = s_1, \ldots, s_{g}$ be the sequence of
lengths of maximal palindromes in $\MaxPalE_T(i)$ sorted in increasing order,
where $g = |\MaxPalE_T(i)|$.
Let $d_j$ be the progression difference for $s_j$,
i.e., $d_j = s_{j} - s_{j-1}$ for $2 \leq j \leq g$.
For convenience, let $d_1 = 0$.
We use the following lemma which is based on
periodic properties of maximal palindromes
ending at the same position.

\begin{lemma}
  \label{lem:maximal_palindromes}
  \hfill
  \begin{enumerate}
    \item[(i)] For any $1 \leq j < g$, $d_{j+1} \geq d_{j}$.
    \item[(ii)] For any $1 < j < g$, if $d_{j+1} \neq d_{j}$, then $d_{j+1} \geq d_j + d_{j-1}$.
    \item[(iii)] $\mathbf{S}_i$ can be represented by $O(\log i)$ arithmetic progressions,
      where each arithmetic progression is a tuple $\langle s, d, t \rangle$ representing the sequence $s, s+d, \ldots, s + (t-1)d$ with common difference~$d$.
    \item[(iv)] If $t \geq 2$, then the common difference $d$ is a period of every maximal palindrome
      which ends at position $i$ in $T$ and whose length belongs to the arithmetic progression $\langle s, d, t \rangle$.
  \end{enumerate}
\end{lemma}
Each arithmetic progression $\langle s, d, t \rangle$ is called
a \emph{group} of maximal palindromes.
Similar arguments hold for the set $\MaxPalB_T(i)$ of
maximal palindromes of $T$ that begin at position~$i$.

To prove Lemma~\ref{lem:maximal_palindromes},
we use arguments from the literature~\cite{Apostolico1995parallel,GasieniecSWAT96,matsubara_tcs2009}.
Let us for now consider any string $W$ of length $m$.
In what follows we will focus on suffix palindromes in $\PalSuf(W)$
and discuss their useful properties.
We remark that symmetric arguments hold for prefix palindromes in $\PalPre(W)$ as well.
Let $\mathbf{S}'= s'_1, \ldots, s'_{g'}$ be the sequence of
lengths of suffix palindromes of $\mathbf{S}'$ sorted in increasing order,
where $g' = |\PalSuf(W)|$.
Let $d'_{j}$ be the progression difference for $s'_{j}$,
i.e., $d'_{j} = s'_{j} - s'_{j-1}$ for $2 \leq j \leq g'$.
For convenience, let $d'_1 = 0$.
Then, the following results are known:

\begin{lemma}[\cite{Apostolico1995parallel,GasieniecSWAT96,matsubara_tcs2009}]
  \label{lem:palindromic_suffixes}
  \hfill
  \begin{enumerate}
    \item[(A)] For any $1 \leq j < g'$, $d'_{j+1} \geq d'_{j}$.
    \item[(B)] For any $1 < j < g'$, if $d'_{j+1} \neq d'_{j}$, then $d'_{j+1} \geq d'_{j} + d'_{j-1}$.
    \item[(C)] $\mathbf{S}'$ can be represented by $O(\log m)$ arithmetic progressions,
      where each arithmetic progression is a tuple $\langle s', d', t' \rangle$ representing the sequence $s', s'+d', \ldots, s' + (t'-1)d'$ of lengths of $t'$ suffix palindromes with common difference $d'$.
    \item[(D)] If $t' \geq 2$, then the common difference $d'$ is a period of every suffix palindrome of $W$ whose length belongs to the arithmetic progression $\langle s', d', t' \rangle$.
  \end{enumerate}
\end{lemma}

The set of suffix palindromes of $W$ whose lengths belong to
the same arithmetic progression $\langle s', d', t' \rangle$ is also called
a \emph{group} of suffix palindromes.
Clearly, every suffix palindrome in the same group has period $d'$,
and this periodicity will play a central role in our algorithms.

We are ready to prove Lemma~\ref{lem:maximal_palindromes}.

\begin{proof}
  It is clear that $\MaxPalE_T(i) \subseteq \PalSuf(T[1..i])$,
  namely,
  \[ \MaxPalE_T(i) = \{s' \in \PalSuf(T[1..i]) \mid T[i-s'] \neq T[i+1], i-s' = 1, \mbox{ or } i = n\}.
  \]

  The case where $i = n$ is trivial,
  and hence in what follows suppose that $i < n$.
  Let $c = T[i+1]$,
  and for a group $\langle s', d', t' \rangle$ of suffix palindromes
  let $a = T[i-s']$ and $b = T[i-s'-(t'-1)d']$,
  namely, $a$ (resp. $b$) is the character that immediately precedes the
  shortest (resp.\ longest) palindrome in the group
  (notice that $a = b$ when $t' = 1$).
  Then, it follows from Lemma~\ref{lem:palindromic_suffixes} (D) that
  $s', s'+d', \ldots, s'+(t'-2)d' \in \MaxPalE_T(i)$
  iff $a \neq c$.
  Also, $s' + (t'-1)d' \in \MaxPalE_T(i)$ iff $b \neq c$.
  Therefore, for each group of suffix palindromes of $T[1..i]$,
  there are only four possible cases:
  (1) all members of the group are in $\MaxPalE_T(i)$,
  (2) all members but the longest one are in $\MaxPalE_T(i)$,
  (3) only the longest member is in $\MaxPalE_T(i)$, or
  (4) none of the members is in $\MaxPalE_T(i)$.

  Now, it immediately follows from Lemma~\ref{lem:palindromic_suffixes} that
  (i) $d_{j+1} \geq d_{j}$ for $1 \leq j < g$ and
  (ii) $d_{j+1} \geq d_{j} + d_{j-1}$ holds for $1 < j < g$.
  Properties (iii) and (iv) also follow from
  the above arguments and Lemma~\ref{lem:palindromic_suffixes}.
\end{proof}

For all $1 \leq i \leq n$ we can compute $\MaxPalE_T(i)$ and $\MaxPalB_T(i)$
in total $O(n)$ time:
After computing all maximal palindromes of $T$ in $O(n)$ time,
we can bucket sort all the maximal palindromes with their ending positions
and with their beginning positions in $O(n)$ time each.

\section{Algorithm for 1-ELPS}
\label{sec:1-ELPS}

In this section, we will show the following result:
\begin{theorem} \label{theo:1_edit_theorem}
  There is an algorithm for the 1-ELPS problem
  which uses $O(n)$ time and space for preprocessing,
  and answers each query in $O(\log (\min \{\sigma, \log n\}))$
  time for single character substitution and insertion,
  and in $O(1)$ time for single character deletion.
\end{theorem}

\subsection{Algorithm for Substitutions}

In what follows, we will present our algorithm to compute
the length of the LPSs after a single
character substitution.
Our algorithm can also return an occurrence of an LPS.

Let $i$ be any position in the string $T$ of length $n$ and let $c = T[i]$.
Also, let $T' = T[1..i-1]c'T[i+1..n]$,
i.e., $T'$ is the string obtained by substituting character $c'$ for
the original character $c = T[i]$ at position $i$.
In this subsection, we assume $c \neq c'$ without loss of generality.
To compute the length of the LPSs of $T'$,
it suffices to consider maximal palindromes of $T'$.
Those maximal palindromes of $T'$ will be computed
from the maximal palindromes of $T$.

The following observation shows that some maximal palindromes
of $T$ remain unchanged
after a character substitution
at position $i$.

\begin{observation}[Unchanged maximal palindromes after a single character substitution]
  \label{obs:key_observation}
  For any position $1 \leq j < i-1$, $\MaxPalE_{T'}(j) = \MaxPalE_T(j)$.
  For any position $i+1 < j \leq n$, $\MaxPalB_{T'}(j) = \MaxPalB_T(j)$.
\end{observation}

By Observation~\ref{obs:key_observation},
for each position $i$~($1 \leq i \leq n$) of $T$,
we precompute the largest element of
$\bigcup_{1 \leq j < i-1} \MaxPalE_T(j)$
and that of
$\bigcup_{i+1 < j \leq n} \MaxPalB_T(j)$,
and store the larger one in the $i$th position of an array
$\mathcal{U}$ of length $n$.
$\mathcal{U}[i]$ is a candidate for the solution after the substitution
at position $i$.
For each position $i$,
$\bigcup_{1 \leq j < i-1} \MaxPalE_T(j)$
contains the lengths
of all maximal palindromes which end to the left of $i$,
and
$\bigcup_{i+1 < j \leq n} \MaxPalB_T(j)$
contains the lengths
of all maximal palindromes which begin to the right of $i$.
Thus, by simply scanning $\MaxPalE_T(j)$ for increasing $j = 1, \ldots, n$
and $\MaxPalB_T(j)$ for decreasing $j = n, \ldots, 1$,
we can compute $\mathcal{U}[i]$ for every position $1 \leq i \leq n$.
Since there are only $2n-1$ maximal palindromes in string $T$,
it takes $O(n)$ time to compute the whole array $\mathcal{U}$.

Next, we consider maximal palindromes of the original string $T$
whose lengths are extended in the edited string $T'$.
As above, let $i$ be the position where
a new character $c'$ is substituted for the original character $c = T[i]$.
In what follows, let $\sigma$ denote
the number of distinct characters appearing in $T$.

\begin{observation}[Extended maximal palindromes after a single character substitution]
  \label{obs:extended_maximal_pals}
  For any $s \in \MaxPalE_T(i-1)$,
  the corresponding maximal palindrome $T[i-s..i-1]$ centered at $\frac{2i-s-1}{2}$ gets extended in $T'$ iff $T[i-s-1] = c'$.
  Similarly, for any $p \in \MaxPalB_{T}(i+1)$,
  the corresponding maximal palindrome $T[i+1..i+p]$ centered at $\frac{2i+p+1}{2}$ gets extended in $T'$ iff $T[i+p+1] = c'$.
\end{observation}
\begin{lemma}\label{lem:one_edit_extended_algorithm}
  Let $T$ be a string of length $n$ over
  an integer alphabet of size polynomial in $n$.
  It is possible to preprocess $T$ in $O(n)$ time and space
  so that later
  we can compute in $O(\log(\min \{\sigma, \log n\}))$ time
  the length of the longest maximal palindromes in
  $T'$ that are extended
  after a substitution
  of a character.
\end{lemma}

\begin{proof}
  By Observation~\ref{obs:extended_maximal_pals},
  we consider maximal palindromes corresponding to\\ $\MaxPalE_T(i-1)$.
  Those corresponding to $\MaxPalB_{T}(i+1)$ can be treated similarly.
  Let $\langle s, d, t \rangle$ be an arithmetic progression
  representing a group of maximal palindromes in $\MaxPalE_T(i-1)$.
  Let us assume that the group contains more than 1 member (i.e., $t \geq 2$)
  and that $i-s \geq 2$,
  since the case where $t = 1$ or $i-s = 1$ is easier to deal with.
  Let $P_j$ denote the $j$th shortest member of the group,
  i.e., $P_1 = T[i-s..i-1]$ and $P_t = T[i-s-(t-1)d..i-1]$.
  Then, it follows from Lemma~\ref{lem:maximal_palindromes} (iv)
  that if $a$ is the character immediately preceding the occurrence of $P_1$
  (i.e., $a = T[i-s-1]$), then $a$ also immediately precedes
  the occurrences of $P_2, \ldots, P_{t-1}$.
  Hence, by Observation~\ref{obs:extended_maximal_pals},
  $P_j$~($2 \leq j < t$) gets extended in the edited text $T'$ iff $c' = a$.
  Similarly, $P_t$ gets extended iff $c' = b$,
  where $b$ is the character immediately preceding the occurrence of $P_t$.
  For each $1 \leq j \leq t$
  the final length of the extended maximal palindrome
  can be computed in $O(1)$ time by a single outward LCE query
  $\OutLCE(i-s-(j-1)d-2, i+1)$.
  Let $P'_j$ denote the extended maximal palindrome
  for each $1 \leq j \leq t$.
  Since there are only $2n-1$ maximal palindromes in string $T$ and all of them can be computed in $O(n)$ total time.

  The above arguments suggest that for each group of
  maximal palindromes,
  there are at most two distinct characters
  that can extend those palindromes
  after a single
  character substitution.
  For each position $i$ in $T$,
  let $\Sigma_i$ denote the set of characters which can extend
  maximal palindromes w.r.t. $\MaxPalE_T(i-1)$
  after a character substitution
  at position $i$.
  It now follows from Lemma~\ref{lem:maximal_palindromes}
  and from the above arguments that $|\Sigma_i| = O(\min \{\sigma, \log i\})$.
  Also, when any character in $\Sigma \setminus \Sigma_i$ is given
  for character substitution at position $i$,
  then no maximal palindromes w.r.t. $\MaxPalE_T(i-1)$ are extended.

  For each maximal palindrome $P$ of $T$,
  let
  $(i', c, l)$
  be a tuple such that
  $i'$
  is the ending position of $P$,
  and $l$ is the length of the extended maximal palindrome $P'$
  after the immediately following character
  $T[i'+1]$
  is substituted for the character
  $c = T[i'-|P|-1]$
  which immediately precedes the occurrence of $P$ in $T$.
  We then radix-sort the tuples
  $(i', c, l)$
  for all maximal palindromes in $T$ as 3-digit numbers.
  This can be done in $O(n)$ time since $T$ is over
  an integer alphabet of size polynomial in $n$.
  Then, for each position
  $i'$,
  we compute the maximum value $l_c$ for each character $c$.
  Since we have sorted the tuples
  $(i', c, l)$,
  this can also be done in total $O(n)$ time for all positions and characters.

  Let $\hat{c}$ be a special character
  which represents any character in $\Sigma \setminus \Sigma_i$
  (if $\Sigma \setminus \Sigma_i \neq \emptyset$).
  Since no maximal palindromes w.r.t. $\MaxPalE_T(i-1)$ are extended
  by $\hat{c}$,
  we associate $\hat{c}$ with the length $\l_{\hat{c}}$
  of the longest maximal palindrome w.r.t. $\MaxPalE_T(i-1)$.
  For each position $i$ and each character $c \in \Sigma_i$, if $l_c$ is less than $l_{\hat{c}}$, then we rewrite $l_c = l_{\hat{c}}$.
  We assume that $\hat{c}$ is lexicographically larger than
  any characters in $\Sigma_i$.
  For each position $i$ we store pairs $(c, l_c)$ in an array $\mathcal{E}_i$ of size
  $|\Sigma_i|+1 = O(\min \{\sigma, \log i\})$ in lexicographical order of $c$.
  See Figure~\ref{fig:one_edit_extended_algorithm}
  for a concrete example.

  \begin{figure}[tb]
    \centerline{
      \includegraphics[scale=0.7]{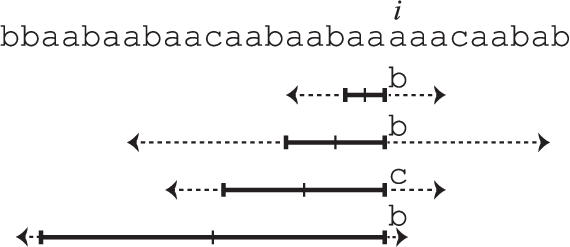}
    }
    \caption{Example for Lemma~\ref{lem:one_edit_extended_algorithm},
      with string $\mathtt{bbaabaabaacaabaabaaaaacaabab}$ where the character $\mathtt{a}$ at position $i = 20$ is to be substituted. There are four maximal palindromes ending at position $19$, whose lengths are represented by two groups $\langle 2, 3, 3 \rangle$ and $\langle 17, 9, 1 \rangle$. For the first group, $\mathtt{c}$ precedes the longest maximal palindrome and $\mathtt{b}$ precedes all the other maximal palindromes. The second group contains only one maximal palindrome and $\mathtt{b}$ precedes it. The largest extended lengths are 21 for $\mathtt{b}$, and 14 for $\mathtt{c}$.
      Thus we have $\mathcal{E}_i = [(\mathtt{b}, 21), (\mathtt{c}, 17), (\hat{c}, 17)]$,
      where $17$ is the length of the longest maximal palindrome ending at position $19$ in the original string.
    }\label{fig:one_edit_extended_algorithm}
  \end{figure}

  Then, given a character $c'$ to substitute for the character
  at position $i$~($1 \leq i \leq n$), we can binary search $\mathcal{E}_i$
  for $(c', l_{c'})$ in $O(\log (\min \{\sigma, \log n\}))$ time.
  If $c'$ is not found in the array,
  then we take the pair $(\hat{c}, l_{\hat{c}})$ from the last entry of $\mathcal{E}_i$.
  We remark that $\sum_{i=1}^n |\mathcal{E}_i| = O(n)$
  since there are $2n-1$ maximal palindromes in $T$
  and for each of them at most two distinct characters contribute to
  $\sum_{i=1}^n |\mathcal{E}_i|$.
\end{proof}

By using hashing instead of binary search,
the query can be solved in $O(1)$ time after $O(n)$ expected time and $O(n)$ space for preprocessing.

Finally, we consider maximal palindromes of the original string $T$
whose lengths are shortened in the edited string $T'$
after substituting a character $c'$ for the original character at position $i$.

\begin{observation}[Shortened maximal palindromes after a single character substitution]
  \label{obs:shortened_maximal_pals}
  A maximal palindrome $T[b..e]$ of $T$ gets
  shortened in $T'$ iff $b \leq i \leq e$ and $i \neq \frac{b+e}{2}$.
\end{observation}

\begin{lemma}\label{lem:shortened_maximal_pals}
  It is possible to preprocess
  a string $T$ of length $n$ in $O(n)$ time and space so that later
  we can compute in $O(1)$ time
  the length of the longest maximal palindromes of $T'$
  that are shortened
  after a substitution
  of a character.
\end{lemma}

\begin{proof}
  Now we consider shortened maximal palindromes whose center $\frac{b+e+1}{2}$ is less than $i$.
  Shortened maximal palindromes whose center $\frac{b+e+1}{2}$ is more than $i$ can be treated similarly.
  Let $\mathcal{S}$ be an array of length $n$ such that
  $\mathcal{S}[i]$ stores the length of the longest maximal palindrome
  after shortening
  by the character substitution at position $i$.
  To compute $\mathcal{S}$,
  we preprocess $T$ by scanning it from left to right.
  Suppose that we have computed $\mathcal{S}[i]$.
  By Observation~\ref{obs:shortened_maximal_pals},
  we have that $\mathcal{S}[i] = 2(i-\frac{b+e+1}{2})$ where
  $T[b..e]$ is the longest maximal palindrome
  of $T$ satisfying the conditions of Observation~\ref{obs:shortened_maximal_pals}.
  In other words, $T[b..e]$ is the maximal palindrome
  of $T$ of which the center $\frac{b+e}{2}$ is the smallest possible
  under the conditions.

  For any position $i < i'' \leq e$,
  we have that
  $\mathcal{S}[i''] = \mathcal{S}[i] + 2(i''-i)$.
  For the next position $e+1$,
  we can compute $\mathcal{S}[e+1]$ in amortized $O(1)$ time
  by simply scanning the array $\mathcal{M}$ from position $\frac{b+e+1}{2}$
  to the right until finding the first (i.e., leftmost) entry of $\mathcal{M}$
  which stores the length of a maximal palindrome
  whose ending position is at least $e+1$.
  Hence, we can compute $\mathcal{S}$ in $O(n)$ total time and space.
\end{proof}

Remark that maximal palindromes of $T$ which do not satisfy
the conditions of Observations~\ref{obs:extended_maximal_pals} and
\ref{obs:shortened_maximal_pals} are also unchanged in $T'$.
The following lemma summarizes this subsection:
\begin{lemma}
  Let $T$ be a string of length $n$ over an integer alphabet of
  size polynomial in $n$.
  It is possible to preprocess $T$ of length $n$
  in $O(n)$ time and space
  so that later we can compute in $O(\log (\min \{\sigma, \log n\}))$ time
  the length of the LPSs of
  the edited string $T'$
  after a substitution of a character.
\end{lemma}

\subsection{Algorithm for Deletions}

Suppose the character at position $i$ is deleted from the string $T$,
and let $T'_i$ denote the resulting string, namely $T'_i= T[1..i-1]T[i+1..n]$.
Now the RL factorization of $T$ comes into play:
Observe that for any $1 \leq i \leq n$, $T'_i = T'_{\RLFbeg(i)} = T'_{\RLFend(i)}$.
Thus, it suffices for us to consider only the boundaries of the RL factors for $T$.

It is easy to see that an analogue of Observation~\ref{obs:key_observation}
for unchanged maximal palindromes holds, as follows.

\begin{observation}[Unchanged maximal palindromes after a single character deletion]
  \label{obs:key_observation_deletions}
  For any position $1 \leq j < \RLFend(i)-1$, $\MaxPalE_{T'_i}(j) = \MaxPalE_T(j)$.
  For any position $\RLFbeg(i)+1 < j \leq n$,   $\MaxPalB_{T'_i}(j) = \MaxPalB_T(j)$.
\end{observation}
See Figure~\ref{fig:delete_unchanged} for a concrete example  of Observation~\ref{obs:key_observation_deletions}.

\begin{figure}[tbh]
  \centerline{
    \includegraphics[scale=0.75]{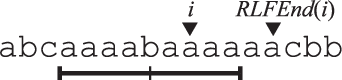}
    \hfil
    \includegraphics[scale=0.75]{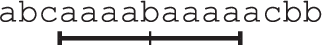}
  }
  \caption{Example for Observation~\ref{obs:key_observation_deletions}. The maximal palindrome $\mathtt{aaaabaaaa}$ does not change if the character $\mathtt{a}$ at position $i$ is deleted. The result is the same if the character $\mathtt{a}$ at position $\RLFend(i)$ is deleted.}
  \label{fig:delete_unchanged}
\end{figure}

By the above observation, we can compute the lengths of the longest unchanged maximal palindromes
for the boundaries of all RL factors in $O(n)$ time,
in a similar way to the case of substitution.

Clearly the new character at position $\RLFend(i)$
in the string
$T'_i$
after a deletion
is always $T[\RLFend(i)+1]$,
and a similar argument holds for $\RLFbeg(i)$.
Thus, we have the following observation for extended maximal palindromes
after a deletion,
which is an analogue of Observation~\ref{obs:extended_maximal_pals}.

\begin{observation}[Extended maximal palindromes after a single character deletion]
  \label{obs:extended_maximal_pals_deletion}
  For any $s \in \MaxPalE_T(\RLFend(i)-1)$,
  the corresponding maximal palindrome $T[\RLFend(i)-s..\RLFend(i)-1]$
  centered at $\frac{2\RLFend(i)-s-1}{2}$ gets extended in $T'_i$ iff $T[\RLFend(i)-s-1] = T[\RLFend(i)+1]$.
  Similarly, for any $p \in \\ \MaxPalB_{T}(\RLFbeg(i)+1)$,
  the corresponding maximal palindrome $T[\RLFbeg(i)+1..\RLFbeg(i)+p]$
  centered at $\frac{2\RLFbeg(i)+p+1}{2}$ gets extended in $T'_i$ iff $T[\RLFbeg(i)+p+1] = T[\RLFbeg(i)-1]$.
\end{observation}
See Figure~\ref{fig:delete_extended} for a concrete example for Observation~\ref{obs:extended_maximal_pals_deletion}.

\begin{figure}[tbh]
  \centerline{
    \includegraphics[scale=0.75]{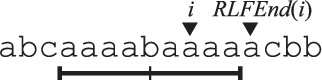}
    \hfil
    \includegraphics[scale=0.75]{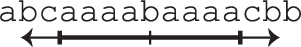}
  }
  \caption{Example for Observation~\ref{obs:extended_maximal_pals_deletion}. The maximal palindrome $\mathtt{aaaabaaaa}$ gets extended to
    $\mathtt{bcaaaabaaaacb}$ if the character $\mathtt{a}$ at position $i$ is deleted. The result is the same if the character $\mathtt{a}$ at position $\RLFend(i)$ is deleted.}
  \label{fig:delete_extended}
\end{figure}

Since the new characters that come from the left and the right
of each deleted position are always unique,
for each $\RLFend(i)$ and $\RLFbeg(i)$,
the longest maximal palindrome that gets extended
after a deletion
is also unique.
Overall, we can precompute their lengths for all positions $1 \leq i \leq n$
in $O(n)$ total time by using $O(n)$ outward LCE queries
in the original string $T$.

Next, we consider those maximal palindromes
which get shortened
after a single
character deletion.
We have the following observation which is analogue to Observation~\ref{obs:shortened_maximal_pals}.

\begin{observation}[Shortened maximal palindromes after a deletion]
  \label{obs:shortened_maximal_pals_deletions}
  A maximal palindrome $T[b..e]$ of $T$ gets
  shortened in $T'_i$ iff $b \leq \RLFbeg(i)$
  and $\RLFend(i) \leq e$.
\end{observation}
See Figure~\ref{fig:delete_shortened} for a concrete example for Observation~\ref{obs:shortened_maximal_pals_deletions}.

\begin{figure}[tbh]
  \centerline{
    \includegraphics[scale=0.75]{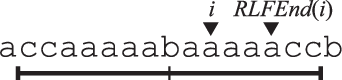}
    \hfil
    \includegraphics[scale=0.75]{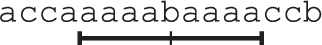}
  }
  \caption{Example for Observation~\ref{obs:shortened_maximal_pals_deletions}. The maximal palindrome $\mathtt{ccaaaaabaaaaacc}$ gets shortened to
    $\mathtt{aaaabaaaa}$ if the character $\mathtt{a}$ at position $i$ is deleted. The result is the same if the character $\mathtt{a}$ at position $\RLFend(i)$ is deleted.}
  \label{fig:delete_shortened}
\end{figure}

By Observation~\ref{obs:shortened_maximal_pals_deletions},
we can precompute the length of the longest maximal palindrome
after deleting the characters at the beginning and ending positions
of each RL factors in $O(n)$ total time,
using an analogous way to Lemma~\ref{lem:shortened_maximal_pals}.

Summing up all the above discussions, we obtain the following lemma:
\begin{lemma}
  It is possible to preprocess a string $T$ of length $n$
  in $O(n)$ time and space
  so that later we can compute in $O(1)$ time
  the length of the LPSs of
  the edited string
  $T'_i$
  after a deletion
  of a character.
\end{lemma}

\subsection{Algorithm for Insertion}

Consider the insertion of a new character $c'$
between the $i$th and $(i+1)$th positions in $T$,
and let $T' = T[1..i]c'T[i+1..n]$.
If $c' \neq T[i]$ and $c' \neq T[i+1]$,
we can find the length of the LPSs
in $T'$ in a similar way to substitution
as follows:
for the maximal palindromes in $T'$ whose center is less than or equal to $i$,
we regard as $c'$ is substituted for $T[i+1]$.
Then, we can compute shortened or unchanged maximal palindromes by using exactly the same algorithm for substitution.
Extended maximal palindromes also can be computed in a similar way to substitution by taking care of the positions of outward LCE queries.
The maximal palindromes in $T'$ whose center is more than or equal to $i+2$ can be computed similarly by regarding as $c'$ is substituted for $T[i]$.
The remaining maximal palindromes in $T'$ with the center $i+0.5$, $i+1$, or $i+1.5$ can be computed easily.
The length of the maximal palindrome in $T'$ with the center $i+1$ is equal to it of the maximal palindrome of $T$ with the center $i+0.5$ plus one.
Also, the maximal palindromes in $T'$ with the center $i+0.5$ or $i+1.5$ are $\varepsilon$.
Otherwise (if $c' = T[i]$ or $c' = T[i+1]$),
we can find the length of the LPSs
in $T'$ in a similar way to deletion
since $c'$ is merged to an adjacent RL factor.
Thus, we have the following.

\begin{lemma}
  Let $T$ be a string of length $n$
  over an integer alphabet of size polynomial in $n$.
  It is possible to preprocess in $O(n)$ time and space
  string $T$ so that later we can compute
  in $O(\log (\min \{\sigma, \log n\}))$ time
  the length of the LPSs of
  the edited string $T'$
  after a insertion
  of a character.
\end{lemma}

\subsection{Hashing}

By using hashing instead of binary searches on arrays,
the following corollary is immediately obtained from Theorem~\ref{theo:1_edit_theorem}.

\begin{corollary}
  There is an algorithm for the 1-ELPS problem
  which uses $O(n)$ \emph{expected} time and $O(n)$ space for preprocessing,
  and answers each query in $O(1)$
  time for single character substitution, insertion, and deletion.
\end{corollary}

\section{Algorithm for $\ell$-ELPS}
\label{sec:block-ELPS}

In this section, we consider the $\ell$-ELPS problem
where an existing block of length $\ell'$ in the string $T$ is replaced with
a new block of length $\ell$.
This generalizes substitution when
$\ell' > 0$ and $\ell > 0$,
insertion when $\ell' = 0$ and $\ell > 0$,
and deletion when $\ell' > 0$ and $\ell = 0$.

This section presents the following result:
\begin{theorem}
  There is an $O(n)$-time and space preprocessing for the $\ell$-ELPS problem
  such that each query can be answered in $O(\ell + \log \log n)$ time,
  where $\ell$ denotes the length of the block
  after an edit.
\end{theorem}
Note that the time complexity for our algorithm is independent
of the length of the original block to edit.
Also, the length $\ell$ of a new block is arbitrary.

Consider
the substitution of a string $X$
of length $\ell$
for the substring $T[i_b..i_e]$ beginning at position $i_b$ and ending at position $i_e$,
where $i_e-i_b+1 = \ell'$ and $X \neq T[i_b..i_e]$.
Let $T'' = T[1..i_b-1]XT[i_e+1..n]$ be the string after the edit.
In order to compute (the lengths of) maximal palindromes
that are affected by the block-wise edit operation,
we need to know the first (leftmost) mismatching position
between $T$ and $T''$,
and that between $\rev{T}$ and $\rev{T''}$.
Let $h$ and $l$ be the smallest integers such that
$T[h] \neq T''[h]$ and $\rev{T}[l] \neq \rev{T''}[l]$, respectively.
If such $h$ does not exist, then let $h = \min\{|T|, |T''|\}+1$.
Similarly, if such $l$ does not exist, then let $l = \min\{|T|, |T''|\}+1$.
Let $j_1 = \lcp(T[i_b..n], XT[i_e..n])$, $j_2=\lcp(\rev{(T[1..i_e])}, \rev{(T[1..i_b]X)})$,
$p_b = i_b+j_1$, and $p_e = i_e-j_2$.
There are two cases:
(1) If $j_1=j_2=0$, then
the first and last characters of $T[i_b..i_e]$ differ from those of $X$.
In this case, we have $i_b = h$ and $i_e = n-l+1$.
We use these positions $i_b$ and $i_e$
to compute maximal palindromes after the block-wise edit.
(2) Otherwise, we have $p_b=i_b+j_1=h$ and $p_e=i_e-j_2=n-l+1$.
We use these positions $p_b$ and $p_e$
to compute maximal palindromes after the block-wise edit.
See Figure~\ref{fig:mismatching position} for illustration.
(2-1) is a sub-case of Case (2) with $p_b(=i_b+j_1) < p_e(=i_e-j_2)$.
In the example of this figure, the substring $T[i_b..i_e]=\mathtt{abbccbabcb}$ is substituted by $X=\mathtt{abbcb}$.
(2-2) is a sub-case of Case (2) with $p_e(=i_e-j_2) < p_b(=i_b+j_1)$.
In the example of this figure, the substring $T[i_b..i_e]=\mathtt{abbcc}$ is substituted by $X=\mathtt{abbbcc}$.
(2-3) is the example $T[i_b..i_e]=\mathtt{abbccbabcb}$ is substituted with $X=\mathtt{abbccb}$ and this is the sub-case of Case (2) with $j_1>\ell$.
Note that $p_b$ and $p_e$ are only used to compute (the lengths of) maximal palindromes and the fact that $T[i_b..i_e]$ is substituted with $X$ is never changed in any case.

In the following,
we describe our algorithm for Case (1).
Case (2) can be treated similarly,
by replacing $i_b$ and $i_e$ with $p_b$ and $p_e$, respectively.
Our algorithm can handle the case where $p_e < p_b$.
Remark that $p_b$ and $p_e$ can be computed in $O(\ell)$ time
by na\"ive character comparisons and a single LCE query each.

\begin{figure}[htpb]
  \centering
    \begin{tabular}{c}

      \begin{minipage}{0.50\hsize}
        \centering
          \includegraphics[keepaspectratio, scale=0.3, angle=0]
                          {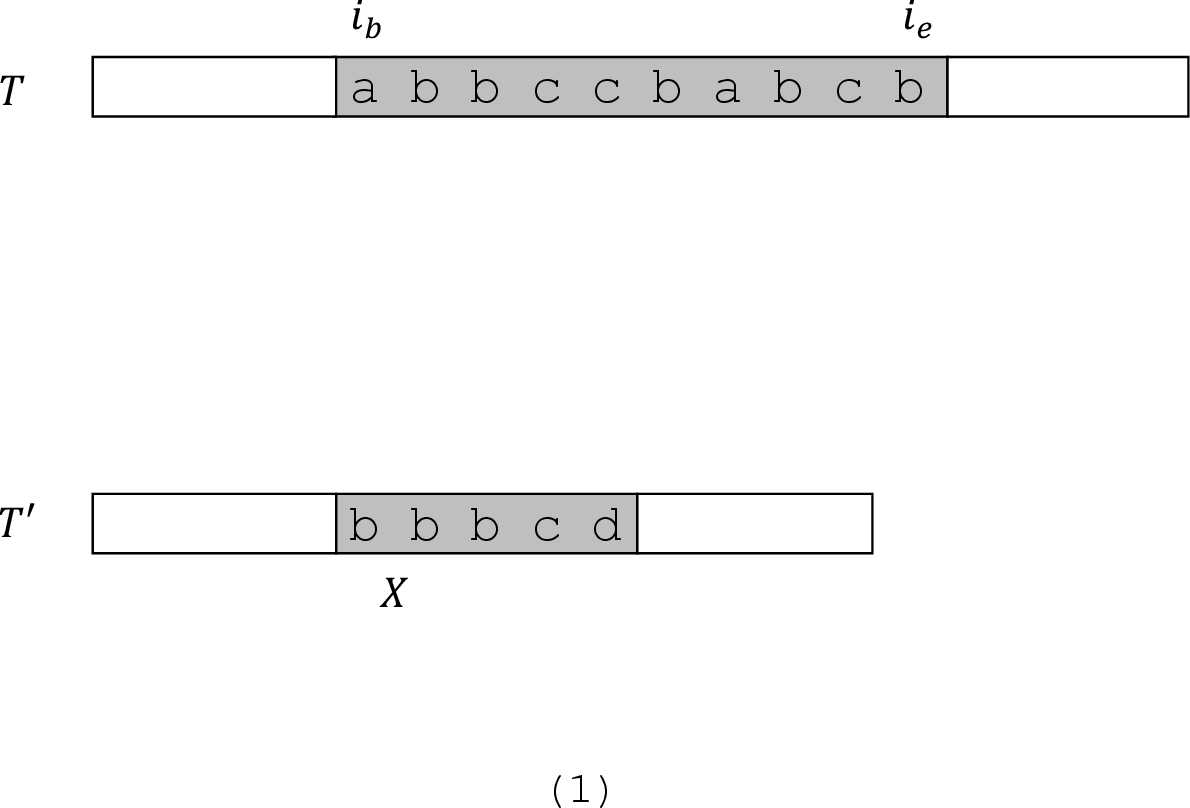}
      \end{minipage}
      \hfill
      \begin{minipage}{0.50\hsize}
        \centering
          \includegraphics[keepaspectratio, scale=0.3, angle=0]
                          {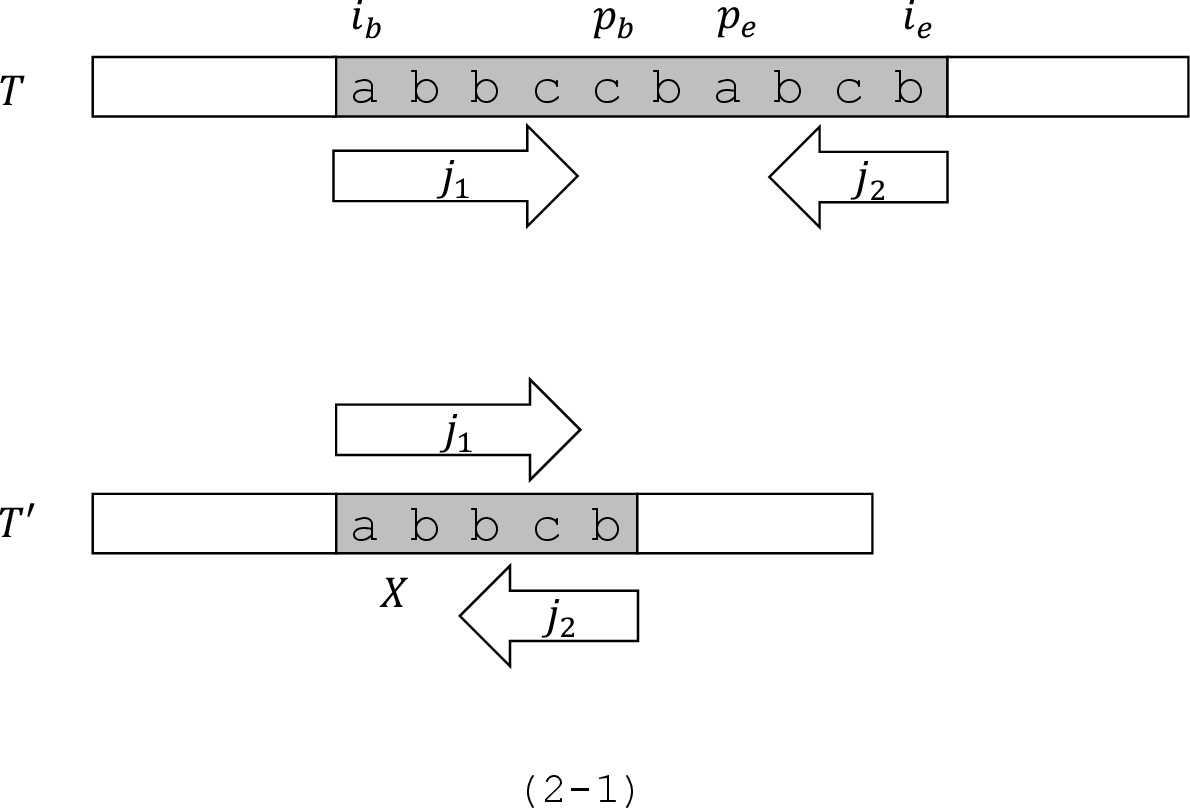}
      \end{minipage} \\

      \begin{minipage}{0.06\hsize}
        \vspace{5mm}
      \end{minipage} \\

      \begin{minipage}{0.50\hsize}
        \centering
          \includegraphics[keepaspectratio, scale=0.3, angle=0]
                          {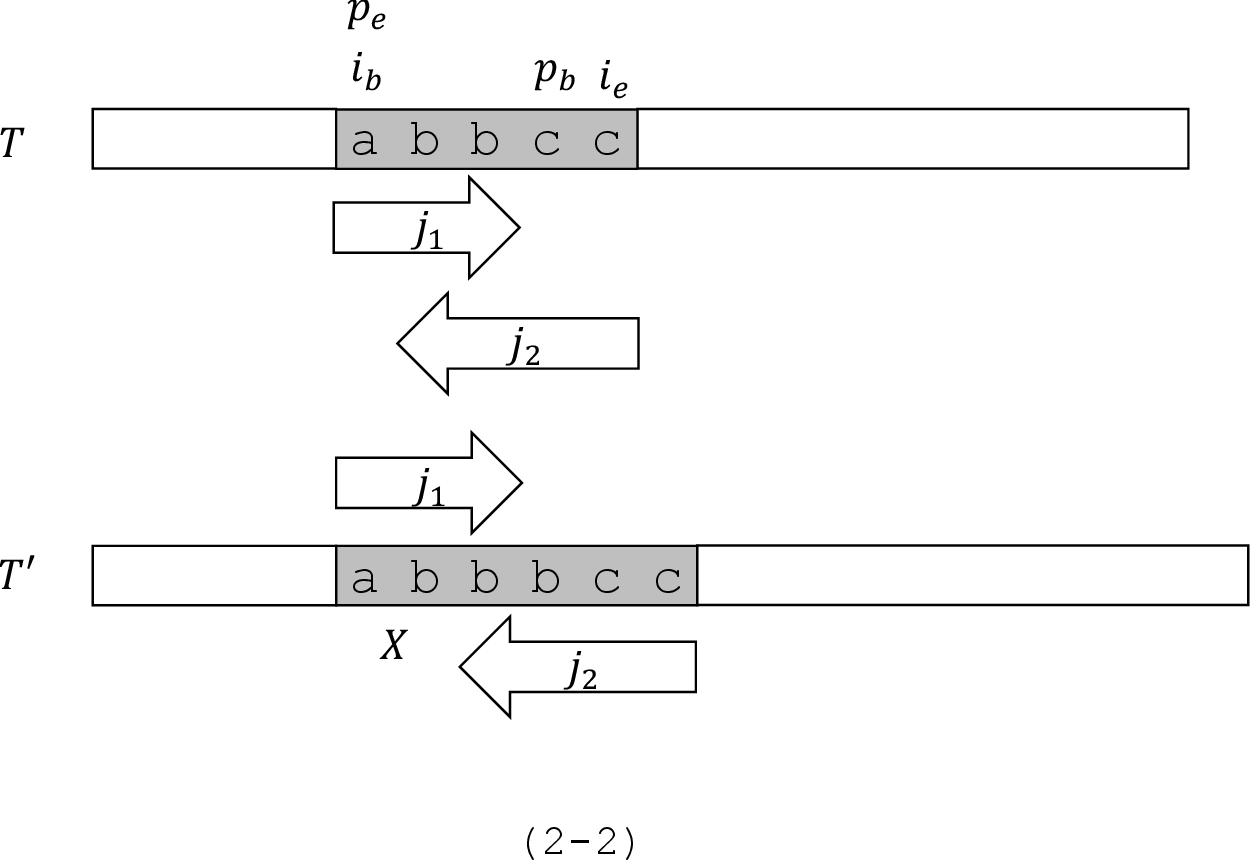}
      \end{minipage}
      \hfill
      \begin{minipage}{0.50\hsize}
        \centering
          \includegraphics[keepaspectratio, scale=0.3, angle=0]
                          {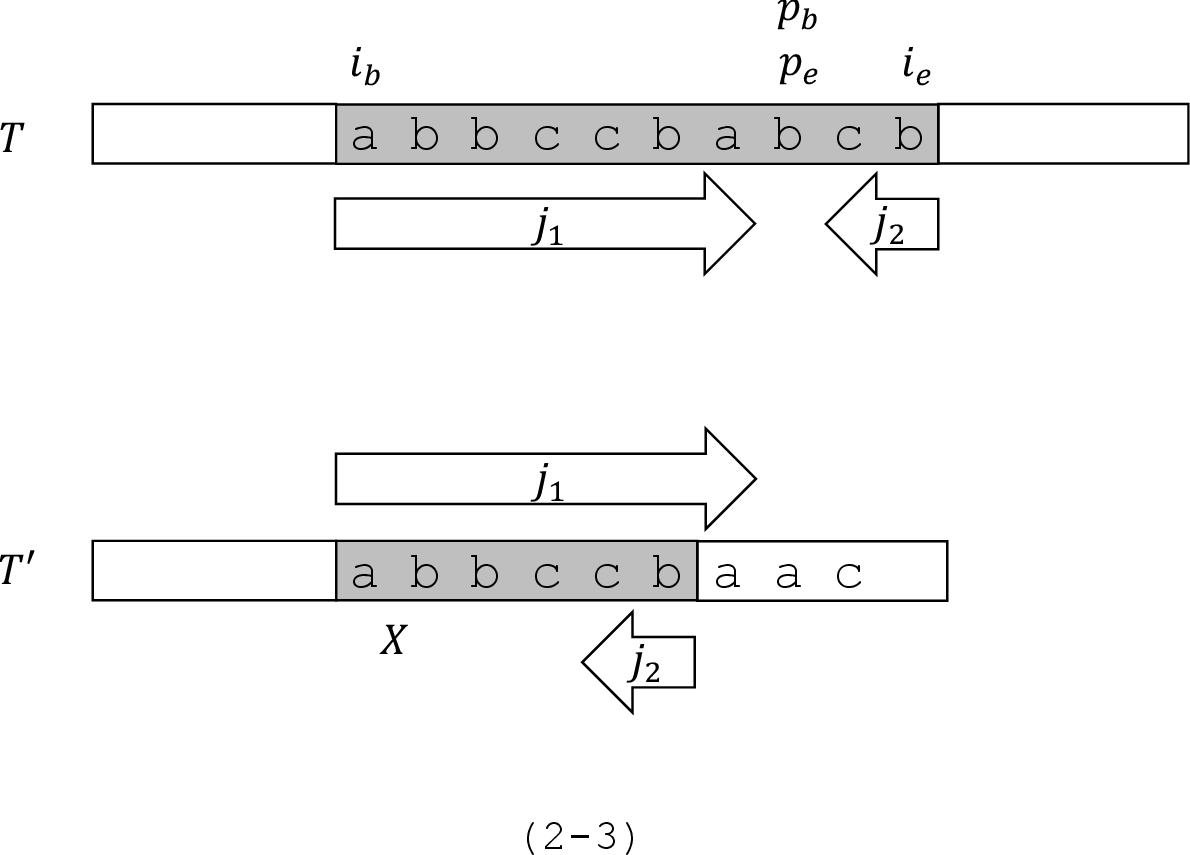}
      \end{minipage}
     \end{tabular}
     \caption{Illustration for the mismatching position between $T$ and $T''$, and that between $T^R$ and $T''^{R}$.
     In particular, (2-1) is the sub-case of Case (2) with $p_b(=i_b+4) < p_e(=i_e-3)$, (2-2) is the sub-case of Case (2) with $p_e(=i_e-4) < p_b(=i_b+3)$, and (2-3) is the sub-case of Case (2) with $j_1(=7)>\ell(=6)$.
     }
     \label{fig:mismatching position}
\end{figure}

\subsection{Unchanged Maximal Palindromes}
We have the following observation
for those of maximal palindromes in $T$ whose
lengths do not change, which is a generalization of Observation~\ref{obs:key_observation}.
\begin{observation}[Unchanged maximal palindromes after a block edit]
  \label{obs:key_observation_block}
  For any position $1 \leq j < i_b-1$, $\MaxPalE_{T''}(j) = \MaxPalE_T(j)$.
  For any position $i_e+1 < j \leq n$, $\MaxPalB_{T''}(j) = \MaxPalB_T(j)$.
\end{observation}
Hence, we can use the same $O(n)$-time preprocessing
and $O(1)$ queries as the 1-ELPS problem:
When we consider substitution for an existing block $T[i_b..i_e]$,
we take the length of the longest maximal palindrome
ending before $i_b-1$ and that of the longest maximal palindrome
beginning after $i_e+1$
as candidates for a solution to the $\ell$-ELPS query.
Hence, we obtain the following lemma.
\begin{lemma}
  We can preprocess a string $T$ of length $n$ in $O(n)$ time and space
  so that later we can compute in $O(1)$ time
  the length of the LPS of $T''$
  that are unchanged
  after a block edit.
\end{lemma}

\subsection{Extended Maximal Palindromes}
Next, we consider the maximal palindromes of $T$
that get extended
after a block edit.
\begin{observation}[Extended maximal palindromes after a block edit]
  \label{obs:extended_maximal_pals_block}
  For any $s \in \MaxPalE_T(i_b-1)$,
  the corresponding maximal palindrome $T[i_b-s..i_b-1]$ centered at $\frac{2i_b-s-1}{2}$ gets extended in $T''$ iff $\OutLCE_{T''}(i_b-s-1, i_b) \geq 1$.
  Similarly, for any $p \in \MaxPalB_{T}(i_e+1)$,
  the corresponding maximal palindrome $T[i_e+1..i_e+p]$ centered at $\frac{2i_e+p+1}{2}$ gets extended in $T''$ iff $\OutLCE_{T''}(i_e, i_e+p+1) \geq 1$.
\end{observation}

\subsubsection{Computation of Extensions}
It follows from Observation~\ref{obs:extended_maximal_pals_block}
that for all
maximal palindromes corresponding to $\MaxPalE_T(i_b-1)$ or $\MaxPalB_{T}(i_e+1)$, it suffices to compute outward LCE queries efficiently in the edited string $T''$.
The following lemma shows how to efficiently compute the extensions of any given maximal palindromes
that end at position $i_b-1$.
Those that begin at position $i_e+1$ can be treated similarly.
\begin{lemma}\label{lem:efficient_batched_extension}
  Let $T$ be a string of length $n$
  over an integer alphabet of size polynomially bounded in $n$.
  We can preprocess $T$ in $O(n)$ time and space so that later,
  given a list of any $f$ maximal palindromes from $\MaxPalE_{T}(i_b-1)$,
  we can compute in $O(\ell + f)$ time
  the extensions of those $f$ maximal palindromes in the edited string $T''$,
  where $\ell$ is the length of a new block.
\end{lemma}

\begin{proof}
  Let us remark that the maximal palindromes in the list can be given
  to our algorithm in any order.
  Firstly, we compute the extensions of given maximal palindromes from the list
  until finding the first maximal palindrome whose extension $\tau$
  is at least one, and let $s'$ be the length of this maximal palindrome.
  Namely, $s'+2\tau$ is the length of the extended maximal palindrome for $s'$,
  and the preceding maximal palindromes (if any) were not extended.
  Let $s$ be the length of the next maximal palindrome from the list after $s'$,
  and now we are to compute the extension $\lambda$ for $s$.
  See also Figure~\ref{fig:suffix_pals_extended}.
\begin{figure}[tb]
  \centering{
    \includegraphics[clip, width=3.0in]{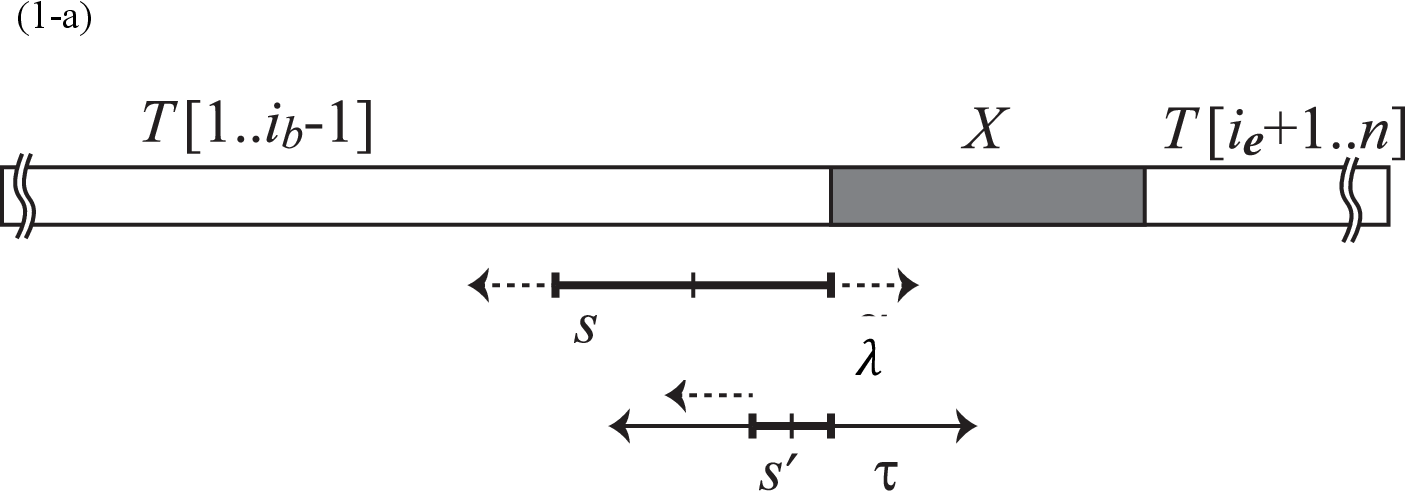}
    \\
    \includegraphics[clip, width=3.0in]{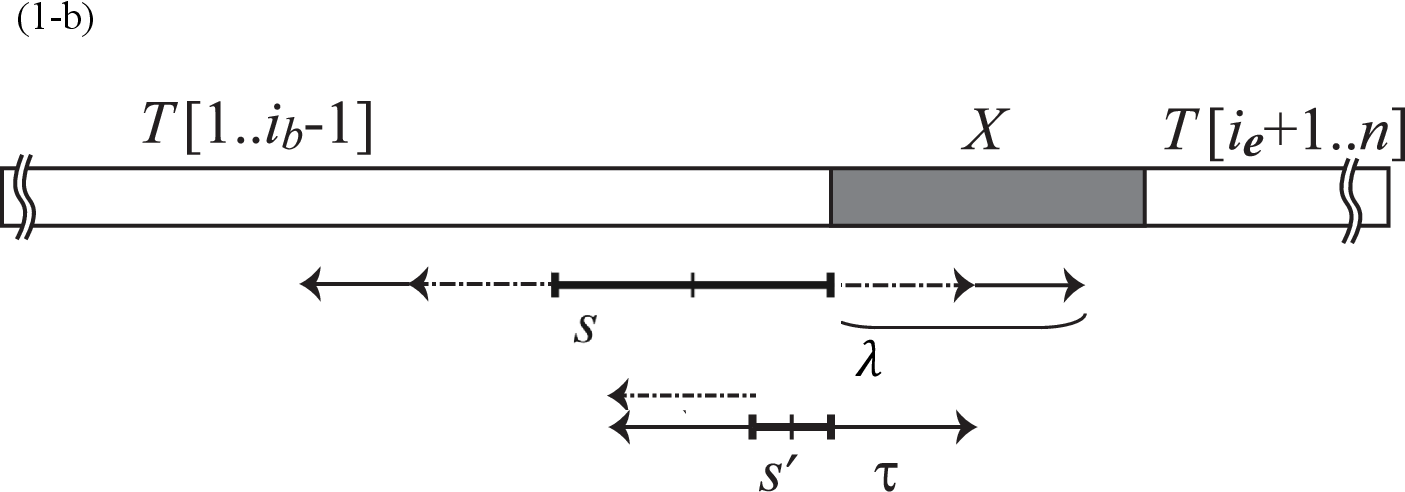}
    \\
    \includegraphics[clip, width=3.0in]{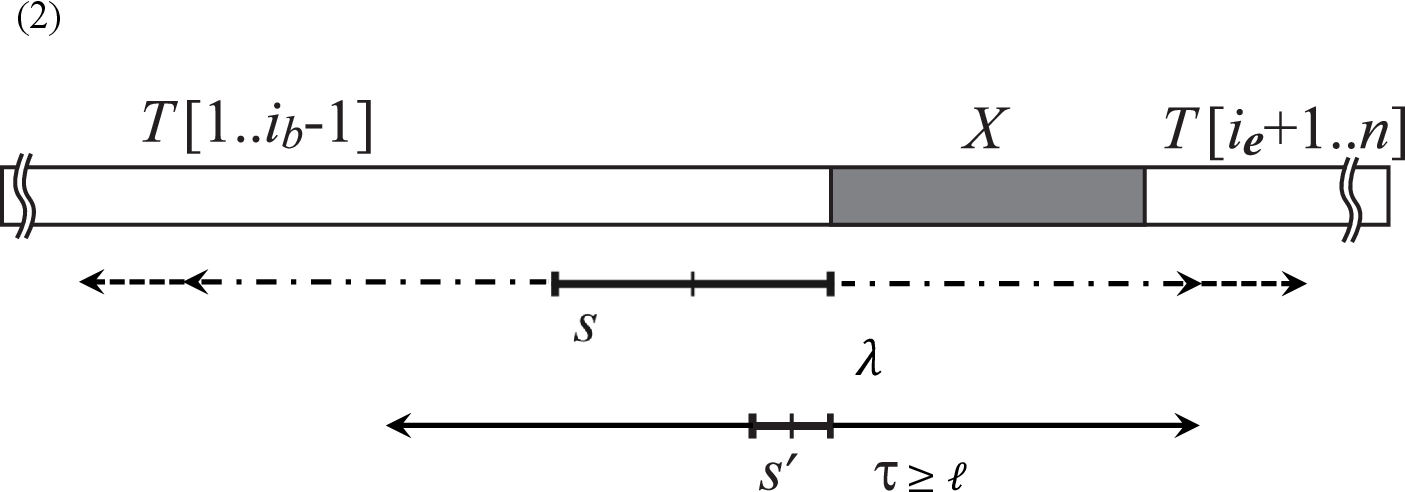}
  }
  \caption{Illustration for Lemma~\ref{lem:efficient_batched_extension},
    where solid arrows represent the matches obtained by
    na\"ive character comparisons,
    and broken arrows represent those obtained by LCE queries.
    This figure only shows the case where $s' < s$,
    but the other case where $s' > s$ can be treated similarly.
}
  \label{fig:suffix_pals_extended}
\end{figure}
  There are two cases:
  (1) If $0 < \tau < \ell$,
  then we first compute $\delta = \LeftLCE_T(i_b-s-1, i_b-s'-1)$.
  We have two sub-cases:
  (1-a) If $\delta < \tau$, then $\lambda = \delta$.
  (1-b) Otherwise ($\delta \geq \tau$), then
  we know that $\lambda$ is at least as large as $\tau$.
  We then compute the remainder of $\lambda$ by na\"ive character comparisons.
  If the character comparison reaches the end of $X$,
  then the remainder of $\lambda$ can be computed by
  $\OutLCE_T(i_b-s-\ell-1, i_e+1)$.
  Then we update $\tau$ with $\lambda$.
  (2) If $\tau \geq \ell$,
  then we can compute $\lambda$ by
  $\LeftLCE_T(i_b-s-1, i_b-s'-1)$,
  and if this value is at least $\ell$, then by
  $\OutLCE_T(i_b-s-\ell-1, i_e+1)$.
  The extensions of the following palindromes can also be computed similarly.

  The following maximal palindromes from the list after $s$ can be processed similarly.
  After processing all the $f$ maximal palindromes in the given list,
  the total number of matching character comparisons is at most $\ell$
  since each position of $X$ is involved in at most one matching character comparison.
  Also, the total number of mismatching character comparisons
  is $O(f)$ since for each given maximal palindrome
  there is at most one mismatching character comparison.
  The total number of LCE queries on the original text $T$
  is $O(f)$, each of which can be answered in $O(1)$ time.
  Thus, it takes $O(\ell + f)$ time to compute
  the length of the $f$ maximal palindromes of $T''$
  that are extended after the block edit.
\end{proof}

However, there can be $\Omega(n)$ maximal palindromes
beginning or ending at each position of a string of length $n$.
In what follows, we show how to reduce the number of maximal palindromes
that need to be considered, by using
periodic structure
of maximal palindromes.

\subsubsection{Longest Extended Palindromes from Each Group}
\label{subsec:longest_extended_palindromes_from_each_group}
Let $\langle s, d, t \rangle$ be an arithmetic progression
representing a group of maximal palindromes ending at position $i_b-1$.
For each $1 \leq j \leq t$,
we will use the convention that $s(j) = s + (j-1)d$,
namely $s(j)$ denotes the $j$th shortest element for $\langle s, d, t \rangle$.
For simplicity, let $Y = T[1..i_b-1]$ and $Z = XT[i_e+1..n]$.
Let $\Ext(s(j))$ denote the length of the maximal palindrome
that is obtained by extending $s(j)$ in $YZ$.

\begin{lemma}\label{lem:batched_extension}
  For any $\langle s, d, t \rangle \subseteq \MaxPalE_T(i_b-1)$,
  there exist palindromes $u, v$ and a non-negative integer $p$,
  such that $(uv)^{t+p-1} u$ (resp. $(uv)^p u$) is
  the longest (resp. shortest) maximal palindrome represented by $\langle s, d, t \rangle$
  with $|uv| = d$.
  Let $\alpha = \lcp(\rev{(Y[1..|Y|-s(1)])}, Z)$
  and $\beta = \lcp(\rev{(Y[1..|Y|-s(t)])}, Z)$.
  Then $\Ext(s(j)) = s(j) + 2\min \{\alpha, \beta + (t-j)d\}$.
  Further, if there exists $s(h) \in \langle s, d, t \rangle$ such that $s(h) + \alpha = s(t) + \beta$,
  then $\Ext(s(h)) = s(h) + 2\lcp(\rev{(Y[1..|Y|-s(h)])}, Z) \geq \Ext(s(j))$
  for any $j \neq h$.
\end{lemma}

\begin{figure}[tb]
  \centerline{
    \includegraphics[scale=0.55]{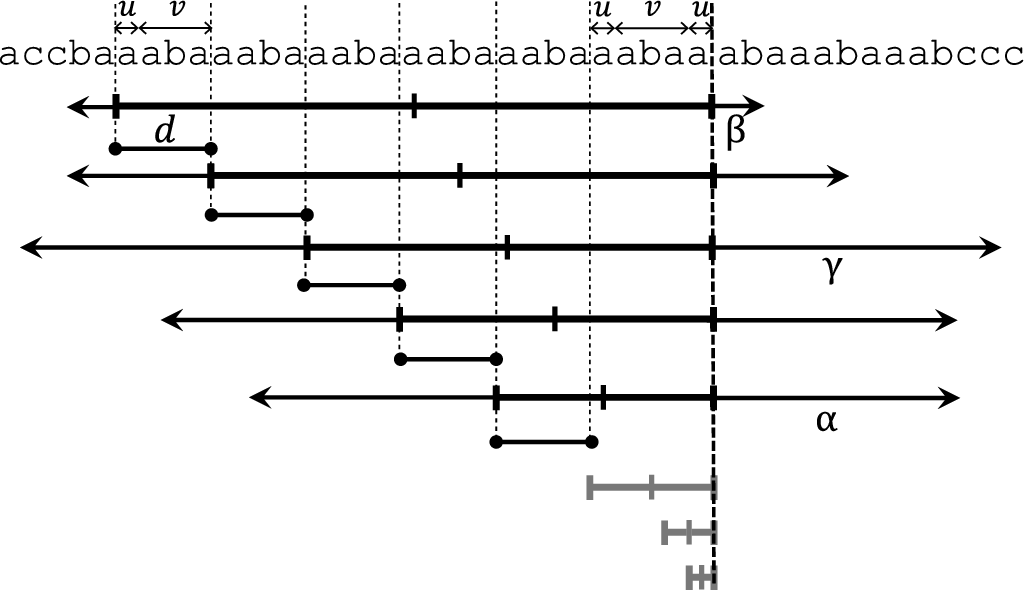}
  }
  \caption{Example for Lemma~\ref{lem:batched_extension},
    where $Y = \mathtt{accbaaabaaabaaabaaabaaabaaabaa}$ and $Z = \mathtt{abaaabaaabccc}$.
    Here $u = \mathtt{a}$ and $v = \mathtt{aba}$.
    The first five maximal palindromes
    $(uv)^pu = (\mathtt{aaba})^p\mathtt{a}$ with $2 \leq p \leq 5$
    belong to the same arithmetic progression (i.e. the same group)
    with common difference $|uv| = d = 4$.
    For this group of maximal palindromes,
    $\alpha = 10$, $\beta = 2$, and $\gamma = 12$.
    Notice that the sixth maximal palindrome $uvu = \mathtt{aabaa}$ belongs to another group
    since the length difference between it and the seventh one $\mathtt{aa}$ is $3$.
  }
    \label{fig:longest_extension_block}
\end{figure}

Then let $\gamma = \lcp(\rev{(Y[1..|Y|-s(h)])}, Z)$.
Lemma~\ref{lem:batched_extension} can be proven immediately
from Lemma 12 of~\cite{matsubara_tcs2009}.
However, for the sake of completeness
we here provide a proof.
We use the following known result:

\begin{lemma}[\cite{matsubara_tcs2009}] \label{lem:matsubara_lem15}
 For any string $Y$ and $\{s(j) \mid s(j) \in \langle s, d, t \rangle\} \subseteq \PalSuf(Y)$, there exist palindromes $u, v$ and a non-negative integer $p$, such that $(uv)^{t+p-1} u$ is a suffix of $Y$, $|uv| = d$ and $|(uv)^p u| = s$.
\end{lemma}
Now we are ready to prove Lemma~\ref{lem:batched_extension}
(see also Figure~\ref{fig:longest_extension_block}).

\begin{proof}
  Let us consider $\Ext(s(j))$, such that $s(j) \in \langle s, d, t \rangle$.
  By Lemma~\ref{lem:matsubara_lem15},
  $Y[|Y|-s(1)-(t-1)d+1..|Y|]=(uv)^{t+p-1} u$,
  where $|uv|=d$ and $|(uv)^p u| = s$.

  Let $x$ be the largest integer such that $(Y[|Y|-x+1..|Y|])^R$ has a period $|uv|$.
  Namely, $(Y[|Y|-x+1..|Y|])^R$ is the longest prefix of $Y^R$ that has a period $|uv|$.
  Then $x$ is given as
  $x = \lcp(\rev{Y}, \rev{(Y[1..|Y|-d])})+d$.
  Let $y$ be largest integer such that $(uv)^{y/d}$ is a prefix of $Z$.
  Then $y$ is given as
  $y = \min\{\lcp(\rev{Y}, Z), x\}$.

  Let $e_l=|Y|-x+1$ and $e_r=|Y|+y$.
  Then, clearly string $T''[e_l..e_r]$ has a period $d$.
  We divide $\langle s, d, t \rangle$ into three disjoint subsets as
  \[
    \langle s, d, t \rangle = \langle s, d, t_1 \rangle \cup \langle s+t_1d, d, t_2 \rangle \cup \langle s+(t_1+t_2)d, d, t_3 \rangle,
  \]
  such that\\
  $|Y|-e_l-s(j)+1 > e_r-|Y|$ for any $s(j) \in \langle s, d, t_1 \rangle$,\\
  $|Y|-e_l-s(j)+1 = e_r-|Y|$ for any $s(j) \in \langle s+t_1d, d, t_2 \rangle$,\\
  $|Y|-e_l-s(j)+1 < e_r-|Y|$ for any $s(j) \in \langle s+(t_1+t_2)d, d, t_3 \rangle$,\\
  $t_1+t_2+t_3=t$, and $t_2 \in \{0,1\}$.

  Then, for any $s(j)$ in the first sub-group $\langle s, d, t_1 \rangle$,
  $\Ext(s(j))=s(j)+2(e_r-|Y|)=s(j)+2y$.
  Also, for any $s(j)$ in the third sub-group $\langle s+(t_1+t_2)d, d, t_3 \rangle$,
  $\Ext(s(j))=s(j)+2(|Y|-e_l-s(j)+1)=s(j)+2(x-s(j))$.
  Now let us consider $s(j) \in \langle a_2, d, t_2 \rangle$,
  in which case $s(j) = s(h)$ (see the statement of Lemma~\ref{lem:batched_extension}).
  Note that $0 \leq t_2 \leq 1$,
  and here we consider the interesting case where $t_2 = 1$.
  Since the palindrome $s(h)$ can be extended beyond the periodicity w.r.t. $uv$,
  we have $\Ext(s(h))=s(h)+2 \gamma$, where $\gamma = \lcp(\rev{(Y[1..|Y|-s(h)])}, Z)$.

  Additionally, we have that
  $y = \lcp(\rev{Y},Z) = \lcp(\rev{(Y[1..|Y|-s(1)])},Z)=\alpha$
  where the second equality comes from the periodicity
  w.r.t. $uv$, and that
  $x-s(j) = \lcp(\rev{(Y[1..|Y|-s(t)])},Z)+(t-j)d=\beta+(t-j)d$.
  Therefore, for any $s(j) \in \langle s, d, t \rangle$,
  $\Ext(s(j))$ can be represented as follows:
  \[
    \Ext(s(j)) = \left\{ \begin{array}{ll}
      s(j)+2\alpha & (\alpha < \beta + (t-j)d ) \\
      s(j)+2 (\beta + (t-j)d) & (\alpha > \beta + (t-j)d )\\
      s(j)+2\gamma & (\alpha = \beta + (t-j)d )
    \end{array} \right.
    \]
  This completes the proof.
\end{proof}

It follows from Lemma~\ref{lem:batched_extension} that
it suffices to consider only three maximal palindromes from each group
(i.e. each arithmetic progression).
Then using Lemma~\ref{lem:efficient_batched_extension},
one can compute the longest maximal palindrome that gets extended
in $O(\ell + \log n)$ time.

\subsubsection{Relationship of Groups}
\label{subsec:relationship_of_groups}
To further speed up computation, we take deeper insights into
combinatorial properties of maximal palindromes in $\MaxPalE_T(i_b-1)$.
Let $G_1, \ldots, G_m$ be the list of all groups for the maximal palindromes
from $\MaxPalE_T(i_b-1)$,
which are
\emph{sorted in increasing order of their common difference}.
Namely, the $j$th shortest member of $\MaxPalE_T(i_b-1)$ belongs to $G_r = \langle s_r, d_r, t_r \rangle$ with $2 \leq r \leq m$, iff the difference between the $j$th shortest maximal palindrome and the $(j-1)$th one is equal to $d_r$.
Then $d_r$ with $2 \leq r \leq m$ is corresponding to the period of any maximal palindrome in $G_r$.
Regardless of whether $\varepsilon$ belongs to $\MaxPalE_T(i_b-1)$ or not, we define that $G_1$ is a singleton, $d_1=0$, and $\varepsilon$ is the element of $G_1$.
When $m = O(\log \log n)$,
$O(\ell+\log\log n)$-time queries immediately follow from Lemmas~\ref{lem:efficient_batched_extension} and~\ref{lem:batched_extension}.
In what follows we consider the more difficult case where
$m = \omega(\log \log n)$. Recall also that $m = O(\log n)$ always holds.

For each $G_r = \langle s_r, d_r, t_r \rangle$ with $1 \leq r \leq m$,
let $\alpha_r$, $\beta_r$, $\gamma_r$, $u_r$, and $v_r$
be the corresponding variables used in Lemma~\ref{lem:batched_extension}.
If there is only a single element in $G_r$, let $\beta_r$ be the length of extension of the palindrome and $\alpha_{r}=\beta_{r-1}$.
For each $G_r$,
let $S_r$ (resp. $L_r$) denote the shortest (resp. longest)
maximal palindrome in $G_r$,
namely, $|S_r| = s_r(1)$ and $|L_r| = s_r(t_r)$.

Each group $G_r$ is said to be of \emph{type-1} (resp. \emph{type-2})
if $\alpha_r < d_r$ (resp. $\alpha_r \geq d_r$).

Let $k$~($1 \leq k \leq m$) be the unique integer such that
$G_k$ is the type-2 group where $d_k$ is the largest common difference
among all the type-2 groups.
Additionally, let $G'_k = G_k \cup \{u_kv_ku_k, u_k\}$.
Note that $u_k$ belongs to one of $G_1, \ldots, G_{k-1}$,
and $u_k v_k u_k$ belongs to either $G_k$ or one of $G_1, \ldots, G_{k-1}$.
In the special case where $\alpha_k=\beta_k+td_k$,
the extensions of $u_k$ and $u_kv_ku_k$ can be
longer than the extension of the shortest maximal palindrome in $G_k$
(see Figure~\ref{fig:G_prime_k}
for a concrete example).
Thus, it is convenient for us to treat $G'_k = G_k \cup \{u_kv_ku_k, u_k\}$
as if it is a single group.
We also remark that this set $G'_k$ is defined only for this specific type-2 group $G_k$.

\begin{figure}[h]
  \centerline{
    \includegraphics[scale=0.50]{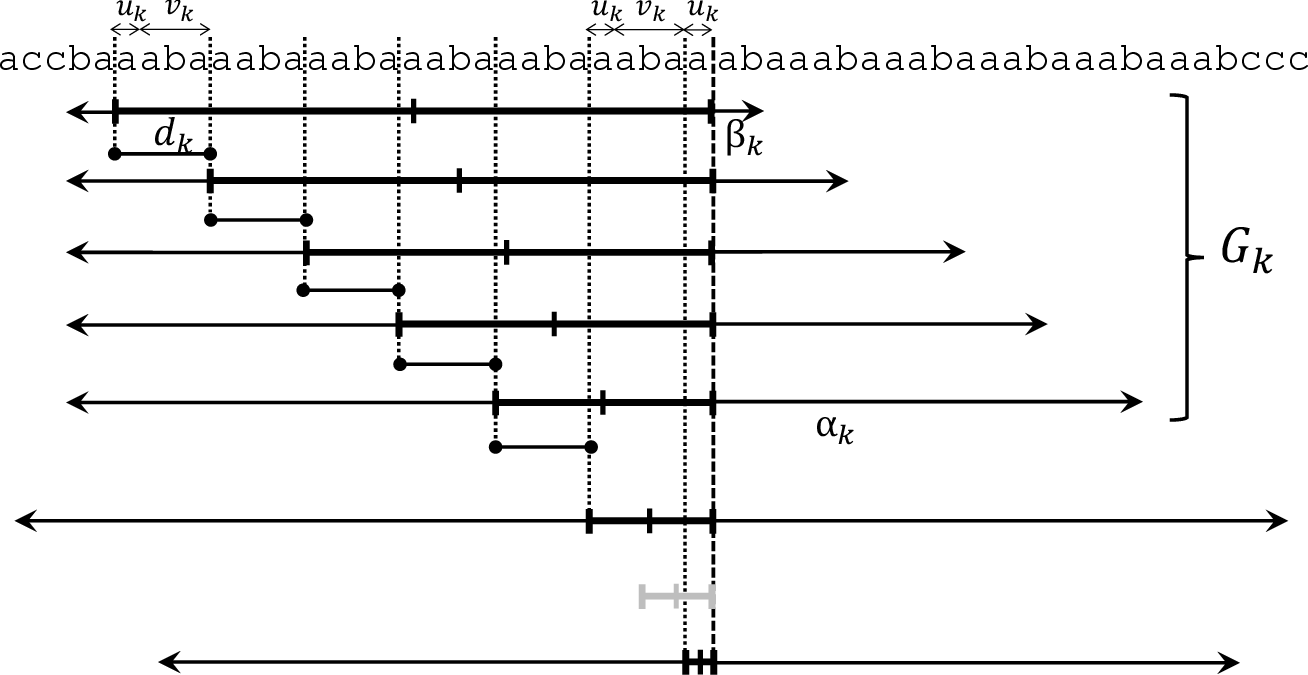}
  }
  \caption{Example for $G'_k = G_k \cup \{u_k v_k u_k, u_k\}$,
  where the extensions of $u_k v_k u_k$ and $u_k$ are longer than
  the extensions of any maximal palindromes in $G_k$.}
    \label{fig:G_prime_k}
\end{figure}

\begin{lemma}\label{lem:group_relationship}
  There is a longest palindromic substring in the edited string $T''$
  that is obtained by extending the maximal palindromes in $G_m$, $G_{m-1}$, or $G'_k$.
\end{lemma}
\begin{proof}
  The lemma holds if the two following claims are true:
  \begin{description}
  \item[Claim (1):]
      The extensions of the maximal palindromes in $G_1, \ldots, G_{k-1}$,
      except for $u_kv_ku_k$ and $u_k$,
      cannot be longer than the extension of the shortest maximal palindrome in $G_k$.

    \item[Claim (2)] Suppose both $G_m$ and $G_{m-1}$ are of type-1. Then, the extensions of the maximal palindromes from $G_{k+1}, \ldots, G_{m-2}$, which are also of type-1, cannot be longer than the extensions of the maximal palindromes from $G_m$ or $G_{m-1}$.
  \end{description}

  \noindent \textbf{Proof for Claim (1).}
  Here we consider the case where the maximal palindrome
  $u_k v_k u_k$ does not belong to $G_k$,
  which implies that the shortest maximal palindrome $S_k$ in $G_k$ is $(u_kv_k)^2u_k$
  (The other case where $u_k v_k u_k$ belongs to $G_k$ can be treated similarly).
  Now, $u_k v_k u_k$ belongs to one of $G_1, \ldots, G_{k-1}$.
  Consider the prefix $P = T[1..i_b-|u_k v_k u_k|-1]$ of $T$
  that immediately precedes $u_k v_k u_k$.
  The extension of $u_k v_k u_k$ is obtained by $\lcp(\rev{P}, Z)$.
  Consider the prefix
  $P' = T[1..i_b-|(u_kv_k)^2u_k|-1]$ of $T$ that immediately precedes $(u_kv_k)^2u_k$.
  It is clear that $P$ is a concatenation of $P'$ and $u_k v_k$.
  Similarly, the prefix $T[1..i_b-|u_k|-1]$ of $T$ that immediately precedes $u_k$ is a concatenation of $P$ and $u_k v_k$.
  It suffices for us to consider only the three maximal palindromes from $G'_k$.
  For any other maximal palindrome $Q$ from $G_1, \ldots, G_{k-1}$,
  assume on the contrary that $Q$ gets extended by at least $d_k$
  to the left and to the right.
  If $|u_k v_k| = d_k < |Q| < |u_k v_k u_k|$,
  then there is an internal occurrence of $u_k v_k$ inside the prefix $(u_k v_k)^2$
  of $(u_k v_k)^2 u_k$.
  Otherwise ($|u_k| < |Q| < |u_k v_k|=d_k$ or $|Q| < |u_k|$),
  there is an internal occurrence of $u_k v_k$ inside $u_k v_k u_k$.
  Here we only consider the first case but other cases can be treated similarly.
  See also Figure~\ref{fig:proof_of_lemma_6-1}.
  This internal occurrence of $u_k v_k$ is immediately followed by
  $u_k v_k w$, where $w$ is a proper prefix of $u_k$ with $1 \leq |w| < |u_k|$.
  Namely, $(u_k v_k)^2w$ is a proper suffix of $(u_k v_k)^2 u_k$.
  On the other hand, $(u_k v_k)^2w$ is also a proper prefix of $(u_k v_k)^2 u_k$.
  Since
  $(u_k v_k)^2 u_k$
  is a palindrome,
  it now follows from Lemma~\ref{lem:pal_border} that $(u_k v_k)^2w$
  is also a palindrome.
  Since $1 \leq |w| < |u_k|$,
  we have $|(u_k v_k)^2w| > |u_kv_ku_k|$
  (note that this inequality holds also when $v_k$ is the empty string).
  Then, $(u_k v_k)^2 w$ is also immediately preceded by $u_k v_k$ because of periodicity and is extended by at least $d_k$
  to the left and to the right
  because $G_k$ is of type-2.
  Since $T''[i_b]=T[i_b-|(u_k v_k)^2 w|-1]$ and $T''[i_b] \neq T[i_b]$, $(u_k v_k)^2 w$ is a maximal palindrome.
  However this contradicts that $(u_k v_k)^2u_k$ belongs to $G_k$ with common difference $d_k = |u_k v_k|$.
  Thus $Q$ cannot be extended by $d_k$ nor more to the left and to the right.
  Since $G_k$ is of type-2, $\alpha_k \geq d_k$.
  Since $|Q| < |(u_k v_k)^2u_k|$,
  the extension of $Q$ cannot be longer than the extension for $(u_k v_k)^2u_k$.
  This completes the proof for Claim (1).

    \begin{figure}[tb]
    \centerline{
      \includegraphics[scale=0.3]{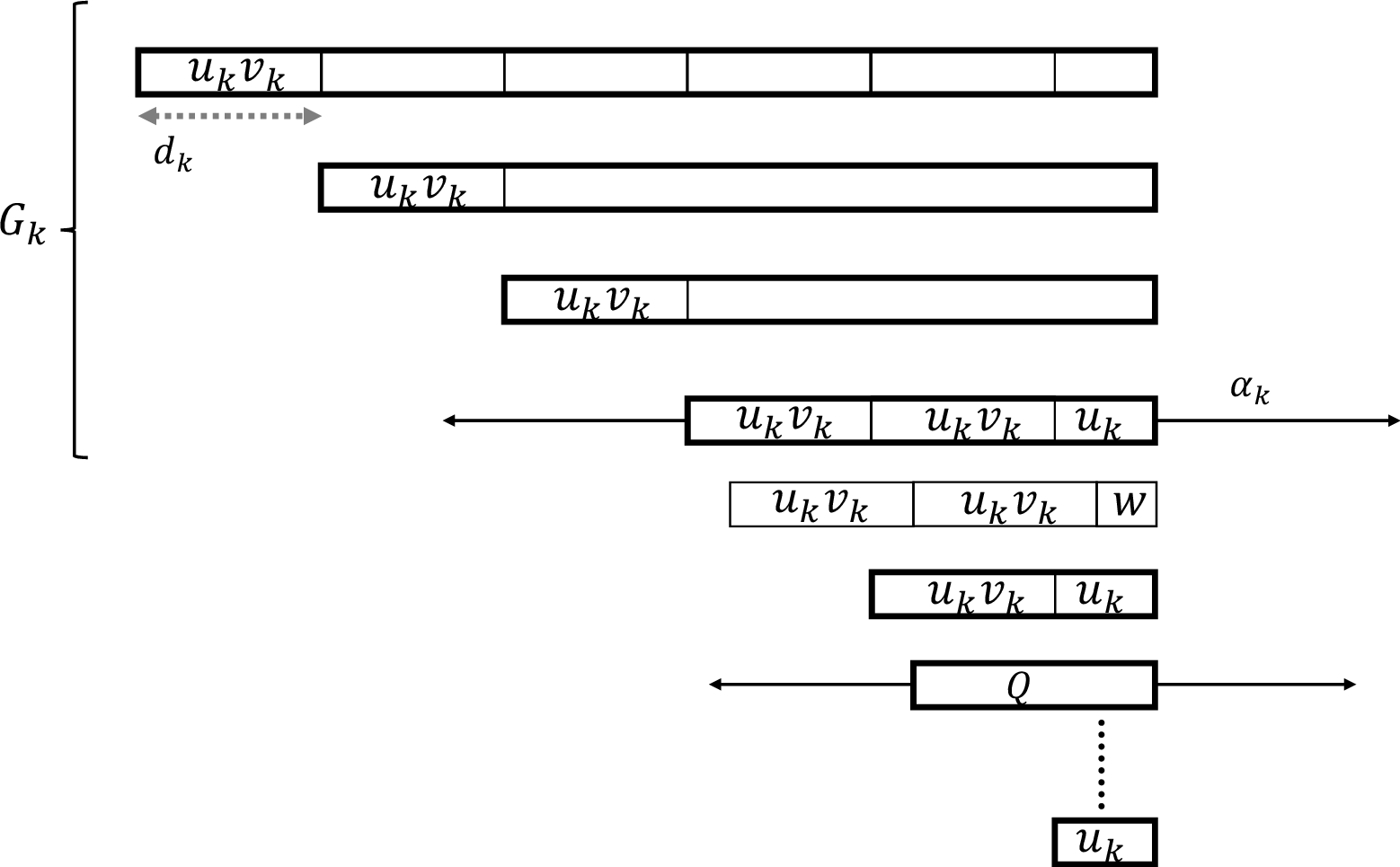}
    }
    \caption{Illustration for the proof for Claim (1) of Lemma ~\ref{lem:group_relationship}.}
    \label{fig:proof_of_lemma_6-1}
  \end{figure}

  \noindent \textbf{Proof for Claim (2).}
  Consider each group $G_r = \langle s_r, d_r, t_r \rangle$ with $k+1 \leq r \leq m-2$.
  By Lemma~\ref{lem:batched_extension},
    $s_r(t_r)+2\beta_r$ and $s_r(t_r-1)+2\alpha_r$ are the candidates
    for the longest extensions of the maximal palindromes from $G_r$.
    Recall that both $G_{m-1}$ and $G_m$ are of type-1,
    and that if $G_r$ is of type-1 then $G_{r+1}$ is also of type-1.
    Now the following sub-claim holds:
  \begin{lemma}
    $\beta_r = \alpha_{r+1}$ for any $k+1 \leq r \leq m-2$.
  \end{lemma}
  \begin{proof}
    If $G_{r+1}$ is a singleton,
    then by definition $\beta_r = \alpha_{r+1}$ holds.
    Now suppose $|G_{r+1}| \geq 2$.
    Since the shortest maximal palindrome $S_{r+1}$ from $G_{r+1}$
    is either $(u_{r+1} v_{r+1})^2 u_{r+1}$ or $u_{r+1} v_{r+1} u_{r+1}$,
    the longest maximal palindrome $L_{r}$ from $G_r$
    is either $u_{r+1} v_{r+1} u_{r+1}$ or $u_{r+1}$.
    The prefix $T[1..i_b-|L_{r}|-1]$ of $T$
    that immediately precedes $L_r$ contains $u_{r+1}v_{r+1}$ as a suffix,
    which alternatively means $\rev{(u_{r+1}v_{r+1})}$ is a prefix of $\rev{(T[1..i_b-|L_r|-1])}$.
    Moreover, it is clear that the prefix $T[1..i_b-|S_{r+1}|-1]$ of $T$
    that immediately precedes $S_{r+1}$ contains $u_{r+1}v_{r+1}$ as a suffix since $|G_{r+1}| \geq 2$.
    In addition, $\alpha_{r+1} < d_{r+1}=|u_{r+1}v_{r+1}|$
    since $G_{r+1}$ is of type-1.
    From the above arguments, we get $\beta_r = \alpha_{r+1}$.
  \end{proof}
  Since $\beta_r=\alpha_{r+1}$ and $\alpha_{r+1}<d_{r+1}$, we have $s_r(t_r)+2\beta_r <s_r(t_r)+2d_{r+1}$.
  In addition, $s_r(t_r-1)+2\alpha_r<s_r(t_r-1)+2d_r=s_r(t_r)+d_r$.
  It now follows from $d_r<d_{r+1}$ that $s_r(t_r)+d_r<s_r(t_r)+2d_{r+1}$.
  Since the lengths of the maximal palindromes and
  their common differences are
  arranged in increasing order in the groups $G_{k+1}, \ldots. G_{m-2}$,
  we have that the longest extension from $G_{k+1}, \ldots. G_{m-2}$ is shorter than
  $s_{m-2}(t_{m-2})+2d_{m-1}$. Since $d_{m-1} < d_{m}$, we have
  \[
  s_{m-2}(t_{m-2}) + 2d_{m-1} < s_{m-2}(t_{m-2}) + d_{m-1} + d_m \leq s_{m-1}(t_{m-1}) + d_{m} \leq s_m = s_m(1).
  \]
  This means that the longest extended maximal palindrome from
  the type-1 groups $G_{k+1}$, $\ldots$, $G_{m-2}$ cannot be longer than
  the original length of the maximal palindrome from $G_m$ before the extension.
  This completes the proof for Claim (2).
\end{proof}

It follows from Lemmas~\ref{lem:efficient_batched_extension}, \ref{lem:batched_extension} and~\ref{lem:group_relationship} that given $G_k$, we can compute in $O(\ell)$ time the length of the LPS of $T''$ after the block edit.
What remains is how to quickly find $G_k$,
that has the largest common difference among all the type-2 groups.
Note that a simple linear search from $G_m$ or $G_1$ takes $O(\log n)$ time,
which is prohibited when $\ell = o(\log n)$.
In what follows, we show how to find $G_k$ in $O(\ell + \log \log n)$ time.

\subsubsection{How to Find $G_k$}
Recall that $T[1..i_b-|L_{r-1}|-1]$ which
immediately precedes $S_r$ contains $u_rv_r$ as a suffix.
Thus, $\rev{(u_rv_r)}$ is a prefix of $\rev{(T[1..i_b-|L_{r-1}|-1])}$.
We have the following observation.

\begin{observation}\label{obs:relationship_between_longestLCE_and_G_k}
  Let $W_1 = \rev{(T[1..i_b-1])}$,
  and $W_r = \rev{(T[1..i_b-|L_{r-1}|-1])}$ for $2 \leq r \leq m$.
  Let $W$ be the string such that $\lcp(W_r,Z)$ is the largest
  for all $1 \leq r \leq m$ (i.e. for all groups $G_1, \ldots, G_m$),
  namely, $W = \argmax_{1 \leq r \leq m} \lcp(W_r,Z)$.
  Then $G_k = G_x$ such that
    \begin{enumerate}
      \item[(a)] $\rev{(u_x v_x)}$ is a prefix of $W$,
      \item[(b)] $d_x \leq \lcp(W,Z)$, and
      \item[(c)] $d_x$ is the largest among all groups that satisfy Conditions (a) and (b).
    \end{enumerate}
\end{observation}

Due to Observation~\ref{obs:relationship_between_longestLCE_and_G_k},
the first task is to find $W$.

\begin{lemma} \label{lem:find_W}
  $W$ can be found in $O(\ell + \log \log n)$ time
  after $O(n)$-time and space preprocessing.
\end{lemma}

\begin{proof}
  We preprocess $T$ as follows.
  For each $1 \leq i \leq n$,
  let $G_1, \ldots, G_m$ be the list of groups that represent the maximal palindromes
  ending at position $i$ in $T$.
  Let $W_1 = \rev{(T[1..i])}$ and
  $W_r = \rev{(T[1..i-|L_{r-1}|])}$ for $2 \leq r \leq m$.
  Let $\mathcal{A}_i$ be the \emph{sparse suffix array} of size $m = O(\log i)$ such that
  $\mathcal{A}_i[j]$ stores the $j$th lexicographically smallest string in $\{W_1, \ldots, W_m\}$.
  We build $\mathcal{A}_i$ with the LCP array $\mathcal{L}_i$.
  Since there are only $2n-1$ maximal palindromes in $T$,
  $\mathcal{A}_i$ for all positions $1 \leq i \leq n$
  can easily be constructed in a total of $O(n)$ time from the full suffix array of $T$.
  The LCP array $\mathcal{L}_i$ for all $1 \leq i \leq n$ can also be computed
  in $O(n)$ total time from the LCP array of $T$
  enhanced with a range minimum query (RMQ) data structure~\cite{DBLP:conf/latin/BenderF00}.

  To find $W$, we binary search $\mathcal{A}_{i_b-1}$ for $Z[1..\ell] = X$
  in a similar way to pattern matching on the suffix array with the LCP array~\cite{manber93:_suffix}.
  This gives us the range of $\mathcal{A}_{i_b-1}$ such that
  the corresponding strings have the longest common prefix with $X$.
  Since $|\mathcal{A}_{i_b-1}| = O(\log n)$,
  this range can be found in $O(\ell + \log \log n)$ time.
  If the longest prefix found above is shorter than $\ell$,
  then this prefix is $W$.
  Otherwise, we perform another binary search on this range for $Z[\ell+1..|Z|] = T[i_e+1..n]$,
  and this gives us $W$.
  Here each comparison can be done in $O(1)$ time by an outward LCE query on $T$.
  Hence, the longest match for $Z[\ell+1..|Z|]$ in this range
  can also be found in $O(\log \log n)$ time.
  Overall, $W$ can be found in $O(\ell + \log \log n)$ time.
\end{proof}

\begin{lemma}\label{lem:find_G_k}
  We can preprocess $T$ in $O(n)$ time and space so that later,
  given $W$ for a position in $T$,
  we can find $G_k$ for that position in $O(\log \log n)$ time.
\end{lemma}

\begin{proof}
Let $\mathcal{D}_i$ be an array of size $|\mathcal{A}_i|$
such that $\mathcal{D}_i[j]$ stores the value of $d_r = |u_rv_r|$,
where $W_r$ is the lexicographically $j$th smallest string in $\{W_1, \ldots, W_m\}$.
Let $\mathcal{R}_i$ be an array of size $|\mathcal{A}_i|$
where $\mathcal{R}_i[j]$ stores a sorted list of common differences $d_r = |u_rv_r|$
of groups $G_r$, such that $G_r$ stores maximal palindromes ending at position $i$
and $\rev{(u_r v_r)}$ is a prefix of the string $\mathcal{A}_i[j]$.
Clearly, for any $j$, $\mathcal{D}_i[j] \subseteq \mathcal{R}_i$.
See also Figure~\ref{fig:array_tree} for an example of $\mathcal{R}_i$.

Suppose that we have found $W$ by Lemma~\ref{lem:find_W},
and let $j$ be the entry of $\mathcal{A}_{i_b-1}$
where the binary search for $X$ terminated.
We then find the largest $d_x$ that satisfies Condition (b)
of Observation~\ref{obs:relationship_between_longestLCE_and_G_k},
by binary search on the sorted list of common differences
stored at $\mathcal{R}_{i_b-1}[j]$.
This takes $O(\log \log n)$ time since the list stored at
each entry of $\mathcal{R}_{i_b-1}$ contains at most $|A_{i_b-1}| = O(\log n)$ elements.

We remark however that the total number of elements in
$\mathcal{R}_i$ is $O(\log^2 i)$ since
each entry $\mathcal{R}_i[j]$ can contain $O(\log i)$ elements.
Thus, computing and storing $\mathcal{R}_i$ explicitly
for all text positions $1 \leq i \leq n$ can take superlinear time and space.

Instead of explicitly storing $\mathcal{R}_i$,
we use a tree representation of $\mathcal{R}_i$, defined as follows:
The tree consists of exactly $m = |\mathcal{A}_i|$ leaves and
exactly $m$ non-leaf nodes.
Each leaf corresponds to a distinct entry $j = 1, \ldots, m$,
and each non-leaf node corresponds to a value from $\mathcal{D}_i$.
Each leaf $j$ is contained in a (sub)tree rooted at a node with $d \in \mathcal{D}_i$,
iff there is a maximal interval $[j'..j'']$
such that
$\mathcal{L}_i[j+1] \geq \mathcal{D}_i[j]$ for any $j \in [j'..j'']$.
We associate each node with this maximal interval.
Since we have
defined $d_1 = 0$,
the root stores $0$ and it has at most $\sigma$ children.
See Figure~\ref{fig:array_tree} that illustrates
  a concrete example for $T[1..i_b-1] = dddF_7^4F_6^2F_5^2F_4F_3^3F_2^2F_1^3$
  with $i_b = 3451$, where
\begin{eqnarray}
  F_1 & = &a \nonumber \\
  F_2 & = &\rev{F_1^3}b \nonumber \\
  F_3 & = &\rev{F_1^3}\rev{F_2^2} \nonumber \\
  F_4 & = &\rev{F_1^3}\rev{F_2^2}\rev{F_3^2}F_2 \nonumber \\
  F_5 & = &\rev{F_1^3}\rev{F_2^2}\rev{F_3^3}\rev{F_4}c \nonumber \\
  F_6 & = &\rev{F_1^3}\rev{F_2^2}\rev{F_3^3}\rev{F_4}\rev{F_5}cF_4F_3^2 \nonumber \\
  F_7 & = &\rev{F_1^3}\rev{F_2^2}\rev{F_3^3}\rev{F_4}\rev{F_5^2}\rev{F_6}\rev{F_3^2}\rev{F_4}cF_5F_4F_3^3F_2^2F_1^2. \nonumber
\end{eqnarray}
We remark that $F_7^4F_6^2F_5^2F_4F_3^3F_2^2F_1^3$, $F_7^3F_6^2F_5^2F_4F_3^3F_2^2F_1^3$, \ldots,
$F_1$ are suffix palindromes of $T[1..i_b-1]$.

We can easily construct this tree in time linear in its size $m = |\mathcal{A}_i|$,
in a bottom up manner.
First, we create leaves for all entries $j = 1, \ldots, m$.
Next, we build the tree in a bottom-up manner,
by performing the following operations in decreasing order of $\mathcal{D}_i[j]$.
\begin{enumerate}
\item[(1)]
  Create a new node with $\mathcal{D}_i[j]$,
  and connect this node with the highest ancestor of leaf $j$.
\item[(2)]
  We check $j' < j$ in decreasing order,
  and connect the new node with
  the highest ancestor of leaf $j'$ iff $\mathcal{L}_i[j'+1] \geq \mathcal{D}_i[j]$.
  We skip the interval corresponding to this ancestor,
  and perform the same procedure until we find $j'$ that does not meet the above condition.
  We do the same for $j'' > j$.
\end{enumerate}
Since each node is associated with its corresponding interval in the LCP array,
it suffices for us to check the conditions $\mathcal{L}_i[j'+1] \geq \mathcal{D}_i[j]$
and $\mathcal{L}_i[j''] \geq \mathcal{D}_i[j]$
only at either end of the intervals that we encounter.
Clearly, in the path from the root to leaf $j$,
the values in $\mathcal{R}_j[j]$ appear in increasing order.
Thus, we can find the largest $d_x$ that satisfies Condition (b)
of Observation~\ref{obs:relationship_between_longestLCE_and_G_k},
by a binary search on the corresponding path in the tree.
We augment the tree with
a level ancestor data structure~\cite{berkman94:_findin,bender04:_level_ances_probl},
so that each binary search takes logarithmic time in the tree height,
namely $O(\log \log n)$ time.
The size of the tree for position $i$
is clearly bounded by the number of maximal palindromes
ending at position $i$.
Thus, the total size and construction time for the trees for all positions in $T$ is $O(n)$.
\end{proof}

\begin{figure}[tb]
  \def\@captype{table}
  \begin{minipage}[c]{0.15\hsize}
    \begin{center}
      \begin{tabular}{|c|l|r|r|l|} \hline
        $j$ & $W_{\mathcal{A}_i[1]}, \ldots , W_{\mathcal{A}_i[m]}$ & $\mathcal{D}_i$ & $\mathcal{L}_i$ & $\mathcal{R}_i$ \\ \hline \hline
        1 & $\mathtt{abaaabaaaaaabaaaba}\cdots$ & 689 & - & 0,1,689 \\
        2 & $\mathtt{aabaaabaaaaaabaaab}\cdots$ & 1 & 1 & 0,1 \\
        3 & $\mathtt{aaabaaabaaaaaabaaa}\cdots$ & 223 & 2 & 0,1,11,223 \\
        4 & $\mathtt{aaabaaabaaaaaabaaa}\cdots$ & 0 & 22 & 0,1,11 \\
        5 & $\mathtt{aaabaaabaaaaaabaaa}\cdots$ & 11 & 33 & 0,1,11\\
        6 & $\mathtt{baaabaaaaaabaaabaa}\cdots$ & 4 & 0 & 0,4 \\
        7 & $\mathtt{baaaaaabaaabaaaaaa}\cdots$ & 37 & 4 & 0,4,37 \\
        8 & $\mathtt{caaabaaabaaaaaabaa}\cdots$ & 82 & 0 & 0,82 \\  \hline
      \end{tabular}
    \end{center}
  \end{minipage}
  \hspace{240pt}
  \begin{minipage}[c]{0.15\hsize}
    \centerline{
      \includegraphics[scale=0.35]{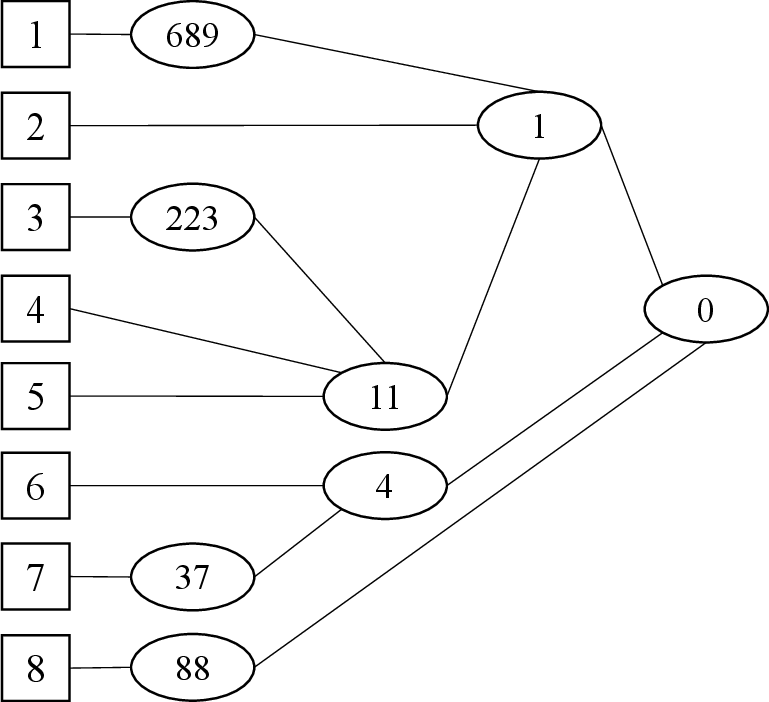}
    }
  \end{minipage}
  \caption{Examples for $\mathcal{R}_i$ (left) and
    its corresponding tree (right).
    The remaining parts of the strings $W_{\mathcal{A}_i[1]}, \ldots , W_{\mathcal{A}_i[m]}$ are omitted due to lack of space.}
  \label{fig:array_tree}
\end{figure}

By Lemmas~\ref{lem:efficient_batched_extension}, \ref{lem:batched_extension},~\ref{lem:group_relationship},~\ref{lem:find_W}, and~\ref{lem:find_G_k}, we can compute in $O(\ell + \log \log n)$ time the length of the LPS of $T''$
that are extended after the block edit.

\begin{remark}
  An alternative method to Lemma~\ref{lem:efficient_batched_extension}
  would be to first build the
  suffix tree of $T\#\rev{T}\$$ enhanced with
  a dynamic lowest common ancestor data structure~\cite{ColeH05}
  using $O(n)$ time and space~\cite{Farach-ColtonFM00},
  and then to update the suffix tree with string $T\#\rev{T}\$X \#' \rev{X}\$'$
  using Ukkonen's online algorithm~\cite{Ukk95},
  where $\#'$ and $\$'$ are special characters not appearing in $T$ nor $X$.
  This way, one can answer LCE queries
  between any position of the original string $T$
  and any position of the new block $X$ in $O(1)$ time.
  Since we need $O(f)$ LCE queries,
  it takes $O(f)$ total time for all LCE queries.
  However, Ukkonen's algorithm requires $O(\ell \log \sigma)$ time
  to insert $X \#' \rev{X} \$'$ into the existing suffix tree,
  where $\ell = |X|$.
  Thus, this method requires us $O(\ell \log \sigma + f)$ time
  and thus is slower by a factor of $\log \sigma$
  than the method of Lemma~\ref{lem:efficient_batched_extension}.
\end{remark}

\subsection{Shortened Maximal Palindromes}
Next, we consider the maximal palindromes that get
shortened
after a block edit.
\begin{observation}[Shortened maximal palindromes after a block edit]
  \label{obs:shortened_maximal_pals_block}
  A maximal palindrome $T[b..e]$ of $T$ gets
  shortened in $T''$ iff $b \leq i_b \leq e$ or $b \leq i_e \leq e$.
\end{observation}
The difference between Observation~\ref{obs:shortened_maximal_pals}
and this one is only in that here we need to consider two positions
$i_b$ and $i_e$.
Hence, we obtain the next lemma
using a similar method to Lemma~\ref{lem:shortened_maximal_pals}:

\begin{lemma}
  We can preprocess a string $T$ of length $n$ in $O(n)$ time and space
  so that later we can compute in $O(1)$ time
  the length of the longest maximal palindromes of $T''$
  that are shortened
  after a block edit.
\end{lemma}

\subsection{Maximal Palindromes whose Centers Exist in the New Block}
Finally, we consider those maximal palindromes
whose centers exist in the new block $X$ of length $\ell$.
By symmetric arguments to Observation~\ref{obs:extended_maximal_pals_block},
we only need to consider
the prefix palindromes and suffix palindromes of $X$.
Using a similar technique to Lemma~\ref{lem:efficient_batched_extension},
we obtain:

\begin{lemma}\label{lem:prefix_pals_in_block}
  We can compute the length of the longest
  maximal palindromes whose centers are inside $X$
  in $O(\ell)$ time and space.
\end{lemma}

\begin{proof}
First, we compute all maximal palindromes in $X$ in $O(\ell)$ time.
Let $p_1, \ldots, p_u$ be a sequence of the lengths of the prefix palindromes
of $X$ sorted in increasing order.
For each $1 \leq j \leq u$,
let $\alpha_j = \lcp(X[p_j+1..\ell], \rev{(T[1..i_b-1])})$,
namely, $p_j + 2\alpha_j$ is the length of the extended maximal palindrome
for each $p_j$.
Suppose we have computed $\alpha_{j-1}$,
and we are to compute $\alpha_j$.
See also Figure~\ref{fig:prefix_pals_in_block}.
\begin{figure}[tb]
  \centerline{
    \includegraphics[scale=0.6]{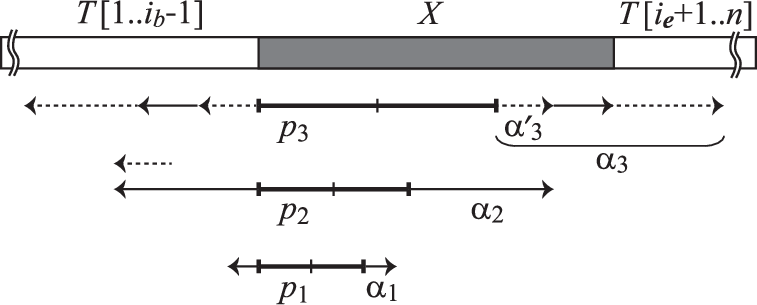}
  }
  \caption{Illustration for Lemma~\ref{lem:prefix_pals_in_block},
    where solid arrows represent the matches obtained by
    na\"ive character comparisons,
    and broken arrows represent those obtained by LCE queries.
    Here are three prefix palindromes of $X$
    of length $p_1$, $p_2$, and $p_3$.
    We compute $\alpha_1$ na\"ively.
    Here, since $p_1+\alpha_1 < p_2$, we compute $p_2$ na\"ively.
    Since $p_2+\alpha_2 > p_3$, we compute
    $\LeftLCE_T(i_b-1, i_b-\alpha_2+\alpha'_3-1)$.
    Here, since its value reached $\alpha'_3$, we perform
    na\"ive character comparison for $X[p_3+\alpha'_3+1..\ell]$
    and $\rev{(T[1..i_b-\alpha'_3-1])}$.
    Here, since there was no mismatch, we perform
    $\OutLCE_T(i_b-\ell+p_3-1, i_e+1)$ and finally obtain $\alpha_3$.
    Other cases can be treated similarly.
  }
  \label{fig:prefix_pals_in_block}
\end{figure}
If $p_{j-1} + \alpha_{j-1} \leq p_{j}$,
then we compute $p_j$ by na\"ive character comparisons.
Otherwise, then let $\alpha'_j = p_{j-1} + \alpha_{j-1} - p_{j}$.
Then, we can compute $\lcp(X[p_j+1..p_j+\alpha'_j], \rev{(T[1..i_b-1])})$
by a leftward LCE query in the original string $T$.
If this value is less than $\alpha'_j$, then it equals to $\alpha_j$.
Otherwise, then we compute $\lcp(X[p_j+\alpha'_j+1..\ell], \rev{(T[1..i_b-1])})$
by na\"ive character comparisons.
The total number of matching character comparisons
is at most $\ell$ since each position in $X$ can be involved
in at most one matching character comparison.
The total number of mismatching character comparisons
is also $\ell$, since there are at most $\ell$ prefix palindromes of $X$
and for each of them there is at most one mismatching character comparison.
Hence, it takes $O(\ell)$ time to compute
the length of the longest maximal palindromes whose centers
are inside $X$.
\end{proof}

\section{Conclusions and Future Work}

In this paper, we dealt with the problems of computing the LPS of a string after a single-character edit operation or a block-wise edit operation.
We proposed an $O(\log (\min \{\sigma, \log n\}))$-time query algorithm that answers the LPS after a single-character edit operation, with $O(n)$-time and space preprocessing,
where $\sigma$ is the number of distinct characters appearing in the string.
Furthermore,
we presented an $O(\ell + \log \log n)$-time query algorithm that answers the LPS after a block-wise edit operation, with $O(n)$-time and space preprocessing,
where $\ell$ denotes the length of the block
after an edit.

Our future work includes the following:
\begin{itemize}
  \item[(1)] Can we efficiently compute the \emph{longest gapped palindrome} in a string after an edit operation? We suspect that it might be possible with a fixed gap length, perhaps using combinatorial properties of gapped palindromes with a fixed gap length from~\cite{IMSIBTNS15}.
  \item[(2)] Can we extend our algorithm to biological palindromes with reverse complements such as those in DNA/RNA sequences?
  The key will be whether or not periodic properties hold for such palindromes.
  \item[(3)] Amir et al.~\cite{AmirBCK19} proposed a fully-dynamic algorithm that can maintain a data structure of $\tilde{O}(n)$ space to report a longest square substring after a single character substitution in $n^{o(1)}$ time. The preprocessing cost for their data structure is $\tilde{O}(n)$ time. It is interesting if one can achieve a faster and/or more space-efficient algorithm for finding a longest square substring, if the edit operation is restricted to a \emph{query} as in this paper.

\end{itemize}

\section*{Acknowledgments}
This work was supported by JSPS KAKENHI Grant Numbers JP20J21147 (MF), JP18K18002 (YN), JP17H01697 (SI), JP16H02783 (HB), JP20H04141 (HB),\\JP18H04098 (MT), and by JST PRESTO Grant Number JPMJPR1922 (SI).

The authors thank anonymous referees for helpful comments, in particular,
for suggesting simpler solutions for Sections~\ref{subsec:longest_extended_palindromes_from_each_group} and \ref{subsec:relationship_of_groups} which are described in Appendix.

\bibliography{ref}

\clearpage
\section{Appendix}\label{sec:appendix}

Section~\ref{subsec:longest_extended_palindromes_from_each_group} and \ref{subsec:relationship_of_groups} can be simplified.

\subsection{Simpler Solution of Section \ref{subsec:longest_extended_palindromes_from_each_group}}

In this subsection, we show that it suffices to consider only one maximal palindrome from each group without using Lemma~\ref{lem:batched_extension}.
The key ideas behind were originally introduced in~\cite{DBLP:conf/mfcs/RubinchikS20}.
First, we compute $e_l$ and $e_r$ for a group $\langle s, d, t \rangle$.
$[e_l..e_r]$ is the maximal interval which has a period $d$.
The maximal palindrome whose center is the closest to $(e_l+e_r)/2$ will be the longest one of the group after edit.
Since the center of the palindrome of length $s(j) \in \langle s, d, t \rangle$ is $i-\frac{s+(j-1)d}{2}$,
the palindrome of length $s(j)$, whose center is the closest to $(e_l+e_r)/2$ among any palindromes in $\langle s, d, t \rangle$, can be found in constant time.
If the center is not equal to $(e_l+e_r)/2$,
then the radius of the longest extended palindrome is the minimum distance from the center to $e_l$ and $e_r$.
Otherwise, the palindrome can be extended beyond the periodicity.
Then the extension can be computed by using an outward LCE query.
Hence, by using Lemma~\ref{lem:efficient_batched_extension},
one can compute the longest maximal palindrome that gets extended
in $O(\ell + \log n)$ time.

\subsection{Simpler Solution of Section \ref{subsec:relationship_of_groups}}

In this subsection, we show proof of the following lemma which is a simplification of Lemma~\ref{lem:group_relationship}.

\begin{lemma}\label{lem:new_group_relationship}
  There is a longest palindromic substring in the edited string $T''$
  that is obtained by extending the maximal palindromes in $G'_m$ or $G'_k$, where $G'_m = G_m \cup \{u_mv_mu_m, u_m\}$.
\end{lemma}

\begin{proof}
  The lemma holds if the two following claims are true:
  \begin{description}
  \item[Claim (1):]
      The extensions of the maximal palindromes in $G_1, \ldots, G_{k-1}$,
      except for $u_kv_ku_k$ and $u_k$,
      cannot be longer than the extension of the shortest maximal palindrome in $G_k$.

    \item[Claim (2)] Suppose $G_m$ is of type-1. Then, the extensions of the maximal palindromes from $G_{k+1}, \ldots, G_{m-1}$, which are also of type-1, cannot be longer than the extensions of the maximal palindromes from $G_m$, except for $u_mv_mu_m$ and $u_m$.
  \end{description}

  \noindent \textbf{Proof for Claim (1).}
  We use the following known result:

  \begin{lemma}[\cite{DBLP:conf/sofsem/KosolobovRS15}] \label{lem:pal^k_lem3}
    Suppose $(uv)^pu$ is a palindrome, where $u$ and $v$ are palindromes and $p$ is a non-negative integer.
    Also suppose $q$ is a palindromic substring of $(uv)^pu$ such that $|q| \geq |uv|-1$.
    Then the center of $q$ coincides with the center of some $u$ or $v$.
  \end{lemma}

  For any maximal palindrome $Q$ from $G_1, \ldots, G_{k-1}$,
  except for $u_kv_ku_k$ and $u_k$,
  assume on the contrary that $Q$ gets extended by at least $d_k$
  to the left and to the right.
  Since $G_k$ is of type-2,
  the interval $[i_b-|Q|-d_k..i_b+d_k]$ is contained in the substring $(u_kv_k)^pu_k$.
  Since $|Q|+2d_k \geq |u_kv_k|-1$,
  the center of $Q$ coincides with the center of some $u_k$ or $v_k$ from Lemma~\ref{lem:pal^k_lem3}.
  However this contradicts that $Q$ does not belongs to $G'_k$.
  Thus $Q$ cannot be extended by $d_k$ nor more to the left and to the right.
  Since $G_k$ is of type-2, $\alpha_k \geq d_k$.
  Since $|Q| < |(u_k v_k)^2u_k|$,
  the extension of $Q$ cannot be longer than the extension for $(u_k v_k)^2u_k$.
  This completes the proof for Claim (1).

  \noindent \textbf{Proof for Claim (2).}
  Consider each group $G_r = \langle s_r, d_r, t_r \rangle$ with $k+1 \leq r \leq m-1$.
  Now we show that the maximal palindromes from $G_r \setminus \{u_{r+1}v_{r+1}u_{r+1}, u_{r+1}\}$ cannot be longer than the original length of the maximal palindrome from $G_{r+1}$ before the extension.
  First, since the longest maximal palindromes from $G_r$ is $u_{r+1}v_{r+1}u_{r+1}$ or $u_{r+1}$, the longest one from $G_r \setminus \{u_{r+1}v_{r+1}u_{r+1}, u_{r+1}\}$ is at least $d_r+d_{r+1}$ symbols shorter than an element of $G_{r+1}$.
  The extension of this palindrome is less than $2d_r$ because $G_r$ is of type-1.
  Therefore, the extended maximal palindromes from $G_r \setminus \{u_{r+1}v_{r+1}u_{r+1}, u_{r+1}\}$ cannot be longer than the element of the previous group.
  This completes the proof for Claim (2).
\end{proof}

\end{document}